\pdfoutput=1 % arXiv: force pdfLaTeX
\documentclass[pdflatex,sn-aps]{sn-jnl}

\usepackage{amsmath}
\usepackage{amssymb}
\usepackage{graphicx}
\usepackage{bm}
\usepackage{microtype}
\usepackage{array}% Flattersatz-Spalten in Tabellen (>{\raggedright...})

\theoremstyle{thmstyleone}
\newtheorem{theorem}{Theorem}
\newtheorem{proposition}[theorem]{Proposition}
\newtheorem{lemma}[theorem]{Lemma}

\theoremstyle{thmstyletwo}

\theoremstyle{thmstylethree}
\newtheorem{remark}[theorem]{Remark}

\newcommand{\EL}[1]{\mathcal{E}_{#1}}

\begin{document}

\title{The Myth of Hamel's Paradox}

\author*{\fnm{Robert} \sur{Singer}}\email{robert.singer@fh-joanneum.at}
\affil{\orgname{FH Joanneum--University of Applied Sciences}, \orgaddress{\city{Graz}, \country{Austria}}}
\affil{Preprint, version 1 (09-2026), DOI: 10.5281/zenodo.22919510}

\abstract{Substituting a velocity constraint into the kinetic energy before forming Lagrange's equations yields, in general, incorrect equations of motion. Recent literature calls this failure ``Hamel's paradox,'' a term introduced in 2010. We show that there is no paradox. The inadmissibility of the substitution was recognized in the 1890s; C.~Neumann (1899) named the operation the \textit{illegitimate form} of the kinetic energy, and Hadamard (1895), Chaplygin (1897), Appell (1899) and Hamel (1904) established its limits of validity. Hamel's textbook of 1949, the nominal source of the paradox, states the prohibition and uses Neumann's term; the result remained in print in the specialist literature and in geometric mechanics. The modern debate cites none of this record; we collect it from the primary sources and measure the debate's claims against it. A geometric analysis identifies the substituted function as the restriction of the mass metric to the constraint distribution. If the distribution is integrable, the restriction is the induced metric of an integral manifold, and the substitution is legitimate. If not, the correctly embedded Boltzmann--Hamel equations differ from those of the restricted energy by a single residue, $\mathrm{Res}_\alpha=c^{s}_{\alpha\beta}\,\omega^\beta P_s$, relative to the complement of the constraint distribution that the shortcut selects. The residue vanishes exactly when the symmetric part of its coefficient vanishes. The two classical admissibility mechanisms, decoupling of the mass metric and vanishing bracket components, are sufficient for this and, for a single constraint, also necessary; for several constraints they are not, which corrects a necessity claim in Hamel's paper of 1904. The fundamental equation of constrained motion---the framework of the debate---is not at issue; what the paper questions is only what the example is taken to show about the foundations.}



\maketitle

\section{Introduction}\label{sec:intro}

Let a mechanical system with kinetic energy $T(q,\dot q)$ be subject to linear velocity constraints. If one substitutes the constraints into $T$ and then forms Lagrange's equations from the reduced function as if the system were unconstrained, the resulting equations of motion are generally wrong. Recent work calls this failure \emph{Hamel's paradox} \cite{udwadia:2010}. Udwadia and Wanichanon introduced the term, presenting it as a paradox posed by Hamel in his 1949 book and stating that Hamel leaves open both why the substitution fails and under which circumstances it does. Chen \cite{chen:2013} and Wanichanon and Cho \cite{wanichanon:2024} adopted the term.

The first purpose of this paper is to point out that there is no paradox. The inadmissibility of the substitution was recognized in the 1890s \cite{vierkandt:1892a,chaplygin:2002,korteweg:1900}, and C.~Neumann gave the operation its classical name, the \emph{illegitimate form} of the kinetic energy \cite{neumann:1899}. Hadamard \cite{hadamard:1895}, Appell \cite{appell:1899} and Hamel \cite{hamel:1904a} established its limits of validity; the Russian literature studies these limits to this day as the Hadamard--Hamel problem \cite{borisov:2015}. Hamel's textbook of 1949, the nominal source of the paradox, states the prohibition in the very example from which the paradox is taken and uses Neumann's term \cite{hamel:1949}. The result then remained in print in the specialist literature \cite{neimark:1972,rosenberg:1977,papastavridis:2002} and in geometric mechanics \cite{bloch:1996,bloch:2005} (Tables~\ref{tab:classical} and~\ref{tab:record}), and its variational counterpart was identified in the same period \cite{flannery:2005}.

The paper that introduced the paradox cites none of this record apart from Hamel's textbook, and the papers that adopted the name took it, together with its historical claims, from that single source. Their questions, and in part their results, are those of the classical literature. In this sense ``Hamel's paradox'' is a myth in the plain meaning of the word: an account of the origin of a result that does not match the record. It attributes to Hamel a question he had answered and presents as paradoxical an operation whose limits of validity had been known for a century.

The second purpose is to give the classical result a form in which it can be stated, proved and taught in a few lines. The substituted function is the restriction of the mass metric to the constraint distribution $D$. If $D$ is integrable, the restriction is the induced metric of an integral manifold, and the substitution is legitimate when it is carried out as the complete passage to that manifold (Proposition~\ref{prop:leaf}). If $D$ is not integrable, the restriction is the constrained Lagrangian of geometric mechanics \cite{bloch:1996,bloch:2005}: a well-defined quadratic form on $D$, which, however, does not determine the dynamics by itself. In the frame calculus of the Boltzmann--Hamel equations, the correctly embedded equations and those formed from the restricted energy differ by one term, the residue $\mathrm{Res}_\alpha=c^{s}_{\alpha\beta}\,\omega^\beta P_s$ built from the structure functions of the adapted frame and the momenta conjugate to the suppressed directions (Proposition~\ref{prop:residue}). The coordinate-level shortcut adds a second, independent term; Proposition~\ref{prop:decomposition} separates the two exactly. The residue vanishes exactly when the symmetric part of its coefficient does. Hadamard's and Hamel's conditions are the two mechanisms that make it vanish direction by direction. They are sufficient; for a single constraint they are also necessary; for several constraints they are not, because the contributions of different constrained directions can cancel (Section~\ref{sec:geo-residue}). Korteweg's exception and Hamel's rule for the frame form follow from the same formula. What is prohibited is thus not the restriction of the energy but the pretense that a coordinate calculus can process it.

The paper addresses the modern debate as a clarification and a complement. The framework in which the debate imposes the constraints, the fundamental equation of constrained motion \cite{udwadia:1996,udwadia:2002}, is not at issue, and its resolution of the example is correct: it is Neumann's rule of 1899. Section~\ref{sec:ill-rediscovery} examines what the 2010 paper infers from the example beyond that rule: that the explanation calls for Gauss's viewpoint, that it dispenses with multipliers, and that it requires singular mass matrices.

The paper is organized as follows. Section~\ref{sec:illegitimate} fixes the operation and its variants, reproduces Hamel's knife-edge example from the source, collects the classical admissibility conditions and documents the record from 1892 to 2015 (Tables~\ref{tab:classical} and~\ref{tab:record}). Section~\ref{sec:ill-rediscovery} measures the claims of the modern debate against this record. Section~\ref{sec:resolution} contains the geometric analysis: the residue and the criterion, the decomposition of the coordinate shortcut and its relation to the constrained Lagrangian, and a comparison with the neighboring lines of work. Section~\ref{sec:discussion} delimits the analysis and concludes. The appendices collect the notation, the proof of the frame identity, Hamel's conditions of 1904 in modern notation, a remark on Chen's criterion and the examples for the criterion.

\section{The operation and the published record}\label{sec:illegitimate}

\subsection{The operation and the two routes}\label{sec:ill-operation}

Let a system be described by coordinates $q^1,\dots,q^n$, kinetic energy $T(q,\dot q,t)$, and generalized applied forces $Q_i$, and let it be subject to $m<n$ Pfaffian constraints
\begin{equation}\label{eq:pfaff}
	B_{r i}(q,t)\,\dot q^i + B_r(q,t) = 0 ,
	\qquad r=1,\dots,m ,
\end{equation}
with the corresponding conditions on the virtual displacements,
\begin{equation}\label{eq:pfaff-virtuell}
	B_{r i}(q,t)\,\delta q^i = 0 .
\end{equation}

The d'Alembert--Lagrange principle states that
\begin{equation}\label{eq:dalembert}
	\bigl(\EL{i}(T) - Q_i\bigr)\,\delta q^i = 0 ,
	\qquad
	\EL{i} := \frac{\mathrm d}{\mathrm dt}\frac{\partial}{\partial\dot q^i} - \frac{\partial}{\partial q^i} ,
\end{equation}
for all $\delta q$ satisfying \eqref{eq:pfaff-virtuell}. Three standing assumptions hold throughout: the constraint matrix $B_{ri}$ has full rank $m$; the kinetic energy is a positive definite quadratic form in the velocities (positive definite mass metric); and the constraints are ideal in d'Alembert's sense---the reaction forces do no work under admissible virtual displacements, which is the content of \eqref{eq:dalembert}. Two routes lead from \eqref{eq:dalembert} to equations of motion.

\emph{Route A (adjoining).} The constraints are adjoined after the variation: by multipliers, $\EL{i}(T) = Q_i + \lambda^r B_{r i}$, together with \eqref{eq:pfaff}; or, equivalently, by contracting the balance $\EL{i}(T)-Q_i$ with a basis of the admissible directions (Maggi's equations \cite{maggi:1901}); or in constraint-adapted quasi-velocities (the Boltzmann--Hamel equations \cite{boltzmann:1902,hamel:1904a}).

\emph{Route B (substitution).} The constraints are substituted into the kinetic energy before the equations of motion are formed. The shortcut occurs in the literature in three variants; their errors differ, and we therefore keep them apart. \emph{(B1) Full elimination:} the constraints \eqref{eq:pfaff} are solved for $m$ velocities, which are eliminated from $T$, producing a function
\begin{equation}\label{eq:Tstar}
	T^{*}\bigl(q,\dot q_{\mathrm f},t\bigr)
	:= T\bigl(q,\dot q(q,\dot q_{\mathrm f},t),t\bigr),
\end{equation}
where $\dot q_{\mathrm f}$ collects the remaining (free) velocities; Lagrange's equations are then formed from $T^{*}$ for the free coordinates as if the system were unconstrained in them; the eliminated coordinates, which may survive as arguments of $T^{*}$, are held fixed. In Neumann's variant of (B1) \cite[p.~436]{neumann:1899} the Lagrange equation of an eliminated coordinate is formed as well. \emph{(B2) Partial substitution:} the constraint combinations are set to zero wherever they appear in $T$, while all coordinates---and, where used, the multipliers---are retained; this is the variant of Hamel's textbook example (Section~\ref{sec:ill-schneide}). \emph{(B3) Premature substitution in quasi-velocities:} the constrained quasi-velocities are set to their constraint values in the frame kinetic energy before the momenta are computed (Sections~\ref{sec:ill-conditions} and~\ref{sec:geo-residue}). ``Route B'' without qualification refers to the family. Figure~\ref{fig:square} shows the two routes as a square that does not commute; Section~\ref{sec:resolution} measures the defect.

\begin{figure}[t]
	\centering
	\includegraphics[width=\textwidth]{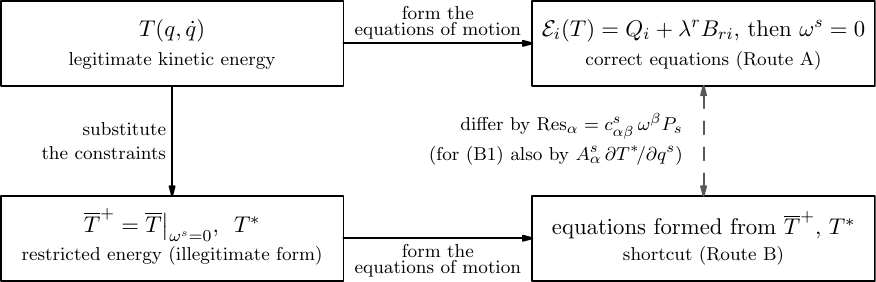}
	\caption{The substitution question as a square that does not commute. Route A forms the equations of motion from the legitimate energy $T$ and imposes the constraints afterwards; Route B substitutes the constraints into $T$ first and forms the equations of the restricted energy $\overline T^{+}$ (in coordinates $T^{*}$). The two results differ by the residue $\mathrm{Res}_\alpha$ of Proposition~\ref{prop:residue}; for full elimination (B1) the coordinate-slot term of Proposition~\ref{prop:decomposition} is added. The square commutes exactly when the constraint distribution is integrable and the embedding is complete, the passage to the integral manifold of Proposition~\ref{prop:leaf}.}\label{fig:square}
\end{figure}

Following Neumann, $T$ is the \emph{legitimate} expression of the kinetic energy and $T^{*}$ its \emph{illegitimate} form \cite{neumann:1899}.\footnote{Neumann, discussing a point mass whose velocity component is controlled by an unknown force: the expression obtained from the Newtonian equations is ``der eigentlich legitime Ausdruck der lebendigen Kraft'' (the properly legitimate expression of the \emph{vis viva}), whereas the substituted expression is ``als durchaus illegitim zur\"uckzuweisen'' (to be rejected as outright illegitimate) \cite[p.~437]{neumann:1899}; the general rule---set up Lagrange's equations ignoring the constraint equations entirely, and take those into account only \emph{a posteriori}---follows there immediately.} Route A is correct; Route B is, in general, not. The classical question, in Hamel's words, is ``wann darf man bei nicht-holonomen Bedingungsgleichungen die Lagrangeschen Gleichungen und die `illegitime Form' der lebendigen Kraft benutzen?''\ \cite[\S 8]{hamel:1904a}: under which circumstances Route B reproduces Route A.

\subsection{Hamel's knife edge}\label{sec:ill-schneide}

Hamel's example is the knife edge (ice-skate blade) moving in the plane \cite[No.~229, pp.~465--466]{hamel:1949}: contact point $B=(x,y)$, blade direction $(\cos\vartheta,\sin\vartheta)$, center of mass at distance $s$ from $B$ along the blade, mass $m$, moment of inertia $I_B=I_S+ms^{2}$ about $B$, with $I_S$ the moment of inertia about the center of mass (Figure~\ref{fig:knife-edge}). The kinetic energy and the no-side-slip constraint are
\begin{equation}\label{eq:schneide-T}
	T = \tfrac12 m\bigl(\dot x^2+\dot y^2\bigr)
	+ m s\,\dot\vartheta\bigl(\dot y\cos\vartheta-\dot x\sin\vartheta\bigr)
	+ \tfrac12 I_B\dot\vartheta^2 ,
	\qquad
	\dot y\cos\vartheta - \dot x\sin\vartheta = 0 .
\end{equation}

\begin{figure}[t]
	\centering
	\includegraphics[scale=0.6]{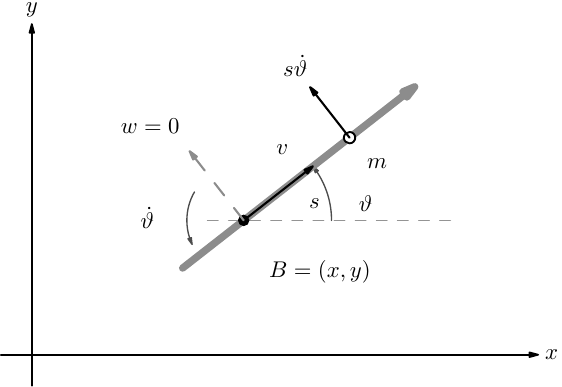}
	\caption{Hamel's knife edge: contact point $B=(x,y)$, blade direction at angle $\vartheta$ to the $x$-axis, center of mass (mass $m$) at distance $s$ from $B$ along the blade. The velocity at $B$ is resolved into the along-blade component $v$ and the transverse component $w=\dot y\cos\vartheta-\dot x\sin\vartheta$; the no-side-slip constraint suppresses the latter, $w=0$ (dashed). The rotation $\dot\vartheta$ gives the center of mass the transverse velocity $s\dot\vartheta$ that the contact point lacks; the offset $s$ thus couples rotation and suppressed direction through the constrained momentum $P_w=ms\dot\vartheta$ of \eqref{eq:schneide-Pw}---the quantity that the illegitimate form discards.}\label{fig:knife-edge}
\end{figure}

Route A (multipliers $\lambda$ for the constraint row $(-\sin\vartheta,\cos\vartheta,0)$, applied force components $X,Y$ and moment $M$) gives
\begin{equation}\label{eq:schneide-korrekt-abc}
	\begin{aligned}
		\frac{\mathrm d}{\mathrm dt}\bigl(m\dot x - m s\sin\vartheta\,\dot\vartheta\bigr) & = X - \lambda\sin\vartheta , \\
		\frac{\mathrm d}{\mathrm dt}\bigl(m\dot y + m s\cos\vartheta\,\dot\vartheta\bigr) & = Y + \lambda\cos\vartheta , \\
		\frac{\mathrm d}{\mathrm dt}\Bigl[m s\bigl(\dot y\cos\vartheta-\dot x\sin\vartheta\bigr)+I_B\dot\vartheta\Bigr]
		+ m s\,\dot\vartheta\bigl(\dot y\sin\vartheta+\dot x\cos\vartheta\bigr)           & = M .
	\end{aligned}
\end{equation}

Using the constraint \emph{after} the equations are formed (which is permitted), and writing $v:=\dot x\cos\vartheta+\dot y\sin\vartheta$ for the speed along the blade, the along-blade combination and the third equation become
\begin{equation}\label{eq:schneide-korrekt}
	m\bigl(\ddot x\cos\vartheta+\ddot y\sin\vartheta\bigr) - m s\,\dot\vartheta^2 = Z ,
	\qquad
	I_B\ddot\vartheta + m s\,\dot\vartheta\,v = M ,
\end{equation}
with $Z:=X\cos\vartheta+Y\sin\vartheta$.\footnote{Hamel's passage contains several misprints: the defective cross-reference discussed below, the momentum $p_\vartheta$ printed as $\partial T/\partial\vartheta$ on p.~465, and a missing dot on $\vartheta$ in the intermediate form of the third equation on p.~466; the final form on p.~467 and the reproduction in \cite[eqs.~(1.8)--(1.12)]{udwadia:2010} carry the dot.} Route B---here Hamel's variant (B2): the constraint combination is set to zero inside $T$, coordinates and multiplier retained---replaces $T$ by
\begin{equation}\label{eq:schneide-Tstar}
	T^{*} = \tfrac12 m\bigl(\dot x^2+\dot y^2\bigr) + \tfrac12 I_B\dot\vartheta^2
\end{equation}
and produces instead
\begin{equation}\label{eq:schneide-falsch}
	m\bigl(\ddot x\cos\vartheta+\ddot y\sin\vartheta\bigr) = Z ,
	\qquad
	I_B\ddot\vartheta = M :
\end{equation}
the centrifugal term $-m s\dot\vartheta^2$ of the offset center of mass and the gyroscopic coupling $m s\dot\vartheta v$ are lost. This is the computation reproduced in \cite{udwadia:2010,chen:2013}.

The structure of the loss is best seen in quasi-velocities. With $v$ as above and the transverse velocity $w:=\dot y\cos\vartheta-\dot x\sin\vartheta$ (the constrained quantity, $w=0$),
\begin{equation}\label{eq:schneide-quasi}
	T = \tfrac12 m\bigl(v^2+w^2\bigr) + m s\,\dot\vartheta\,w + \tfrac12 I_B\dot\vartheta^2 ,
\end{equation}
so $T$ contains the constrained velocity $w$ once quadratically and once \emph{linearly}. The quadratic term is harmless; the linear term carries the momentum conjugate to the suppressed direction,
\begin{equation}\label{eq:schneide-Pw}
	P_w := \frac{\partial T}{\partial w}\bigg|_{w=0} = m s\,\dot\vartheta \neq 0 ,
\end{equation}
and it is exactly this quantity that multiplies the lost terms in \eqref{eq:schneide-korrekt}.

For $s=0$ the linear term vanishes and variant (B2) discards nothing: the equations come out correct, an admissible case of the energy-splitting kind in Hamel's classification (Section~\ref{sec:ill-conditions}). Hadamard's structural mechanism is not available here, since for a single constraint his criterion reduces to integrability \cite[No.~7]{hadamard:1895}. Full elimination (B1) is not rescued by $s=0$: eliminating $\dot y=\dot x\tan\vartheta$ gives $T^{*}=\tfrac12 m\sec^{2}\!\vartheta\,\dot x^{2}+\tfrac12 I_B\dot\vartheta^{2}$, whose force-free Lagrange equations $\frac{\mathrm d}{\mathrm dt}\bigl(m\sec^{2}\!\vartheta\,\dot x\bigr)=0$ and $I_B\ddot\vartheta-m\sec^{2}\!\vartheta\tan\vartheta\,\dot x^{2}=0$ disagree with the correct equations $m\dot v=0$, $I_B\ddot\vartheta=0$. The two variants fail for different reasons; the discrepancy of (B1) is quantified in Section~\ref{sec:geo-decomposition}.

Hamel states the prohibition at the example itself: one may use the constraint after the equations of motion are set up, ``but we would not have been permitted to work beforehand with the expression of the kinetic energy simplified by the constraint equation [\dots]. That would evidently have given false equations''\footnote{``wir h\"atten aber nicht vorher mit dem durch die Bedingungsgleichung vereinfachten Ausdruck der kinetischen Energie [\dots] arbeiten d\"urfen. Das h\"atte ersichtlich falsche Gleichungen gegeben.''} \cite[p.~466]{hamel:1949}. The explanation appears a few pages later in the same chapter \cite[Nos.~237, 240]{hamel:1949} (Section~\ref{sec:ill-conditions}, item~5).\footnote{The cross-reference appended there, ``(vgl.\ Kap.~VI, \S~3)'', is defective: the section so numbered (pp.~295--299) treats first integrals of the canonical equations and Poisson brackets, and no passage of the book unambiguously fits the citation; most probably the reference is a misprint.}

\subsection{The classical admissibility conditions}\label{sec:ill-conditions}

The classical literature did not stop at the prohibition; it determined why the substitution fails and under which conditions it may nevertheless be used.

\begin{enumerate}
	\item \emph{Vierkandt's localization.} The derivation of Lagrange's equations from the fundamental equation is tied to ``a quite specific condition'' (``eine ganz bestimmte Bedingung''). Vierkandt locates it in two places: in his equation (20), the integration by parts of the acceleration terms,
	      \begin{equation*}
		      \sum\ddot x\,\delta x
		      =\frac{\mathrm d}{\mathrm dt}\sum\dot x\,\delta x-\sum\dot x\,\delta\dot x ,
	      \end{equation*}
	      which presupposes the relation $\mathrm d\delta x/\mathrm dt=\delta\,\mathrm dx/\mathrm dt$ (his (a)); and in the relations $\partial\dot x/\partial\dot q_i=\partial x/\partial q_i$ (his (b)). Both presuppose that $\dot x,\dot y,\dot z$ are total time derivatives, which they cease to be once the constraint equations have been used in forming $T$. The substituted $T^{*}$ thus violates the premise of the derivation, and the simplification is ``in general inadmissible'' (``im allgemeinen unstatthaft'') \cite[\S4]{vierkandt:1892a}.
	\item \emph{Hadamard's conditions, Appell's criterion, and Chaplygin's statement.} Hadamard \cite[Nos.~5--8]{hadamard:1895} writes $p$ linear constraints on $m+p$ parameters in solved form and asks which linear combinations of them may be inserted into $T$ before the differentiations:
	      \begin{equation*}
		      \dot q_k=\sum_{h=1}^{m}a^{k}_{h}\,\dot q_h\quad(k=m{+}1,\dots,m{+}p),
		      \qquad
		      \mathcal C=\sum_{k}\lambda_k\Bigl(\sum_{h}a^{k}_{h}\,\dot q_h-\dot q_k\Bigr).
	      \end{equation*}

	      Using the constraint equations in $T$ means, algebraically, replacing $T$ by $T+\mathcal C$, where the $\lambda_k$ are functions of the $q$ and, in general, of the $\dot q$; since $\mathcal C$ vanishes on the constrained velocities, the two functions agree there, but each Lagrange expression changes by the value it takes for $\mathcal C$ alone, and the replacement is legitimate only if that value vanishes in the final result \cite[No.~6]{hadamard:1895}. For multipliers $\lambda_k$ independent of the velocities the added term in the equation of $q_i$ is $\sum_h\dot q_h\,P_{i,h}$, and it vanishes identically exactly when
	      \begin{equation*}
		      P_{i,h}:=\sum_{k}\lambda_k\Bigl[\frac{\partial a^{k}_{i}}{\partial q_h}-\frac{\partial a^{k}_{h}}{\partial q_i}
			      -\sum_{l}\Bigl(a^{l}_{i}\frac{\partial a^{k}_{h}}{\partial q_l}-a^{l}_{h}\frac{\partial a^{k}_{i}}{\partial q_l}\Bigr)\Bigr]=0 ,
		      \qquad P_{i,h}=-P_{h,i}
	      \end{equation*}
	      \cite[No.~7, eqs.~(9)--(10)]{hadamard:1895}: $m(m-1)/2$ linear conditions on the $\lambda_k$ with coefficients from the constraints alone, so that at least $p-m(m-1)/2$ substitutable combinations always exist and $T$ can always be reduced to $m(m+1)/2$ velocities. All constraints are substitutable if and only if the system is integrable \cite[No.~8]{hadamard:1895}. In the frame language of Appendix~\ref{app:hamel-par8}, $P_{i,h}=-\lambda_k\,c^{k}_{ih}$ for the coordinate-adapted frame: a combination is substitutable exactly when its form is exact at each point of the constraint plane \cite[Nos.~10--12]{hadamard:1895}. Rolling examples: two surfaces rolling on one another ($m=3$, $p=2$) admit no substitutable combination; rolling without pivoting ($m=2$, $p=3$) admits exactly the two non-slip equations \cite[Nos.~13--14]{hadamard:1895}; see \cite{borisov:2015} for the modern development.

	      Appell \cite[No.~23]{appell:1899} states the criterion for a \emph{given} system. With the admissible velocities $x'=\sum_\nu a_\nu\,q'_\nu$ (likewise $y'$, $z'$), where the coefficients depend on the retained parameters $q_1,\dots,q_n$ alone, the correct equation for $q_1$ and its deviation from Lagrange's form are
	      \begin{equation*}
		      \frac{\mathrm d}{\mathrm dt}\frac{\partial T}{\partial q_1'}-R_1=Q_1 ,
		      \qquad
		      R_1-\frac{\partial T}{\partial q_1}
		      =\sum m\Bigl[x'\sum_{\nu}\Bigl(\frac{\partial a_1}{\partial q_\nu}-\frac{\partial a_\nu}{\partial q_1}\Bigr)q'_\nu+\dotsb\Bigr]
	      \end{equation*}
	      \cite[eqs.~(45)--(47)]{appell:1899}; the deviation is a quadratic form in the velocities, and Lagrange's equation applies to $q_1$ if and only if this form vanishes identically. It vanishes for all parameters when the displacements are exact differentials, and for $q_1$ alone when $\partial a_1/\partial q_\nu=\partial a_\nu/\partial q_1$ for all $\nu$ (likewise $b$, $c$), in which case $q_1$ is a true coordinate once the remaining parameters are known functions of time \cite[eq.~(48)]{appell:1899}. Hadamard's criterion concerns the combinations substitutable for every kinetic energy and is a property of the constraints alone; Appell's concerns one system at a time.

	      Chaplygin, in the report of 1895 that exposed Lindel\"of's error,\footnote{Lindel\"of formed the kinetic energy with the rolling conditions taken into account and then wrote Lagrange's equations for it, ``not in the form (7), as he ought to do, but in the form (8); this was his mistake'' \cite[p.~121]{chaplygin:2002}.} writes the equations of motion as Lagrange's operator applied to the substituted function, plus correction terms built from the momenta of the eliminated velocities and the curl of the constraint coefficients---the form \eqref{eq:decomposition} without its first term. He states that these equations reduce to Lagrange's form ``in the only case'' that the coefficients satisfy the integrability relations; the exceptional case of vanishing momenta of the eliminated velocities is noted and set aside \cite[pp.~120--121]{chaplygin:2002}. That is Hadamard's criterion together with the energy-splitting case of Hamel's conditions below.
	\item \emph{Hamel's conditions.} Hamel \cite[\S8]{hamel:1904a} states when Lagrange's equations formed from the illegitimate $T$ are valid for the free coordinates: the constrained quasi-velocities must admit a splitting such that no products of free and constrained velocities occur in $T$ and certain structure coefficients vanish. In the extreme cases, either the admissible infinitesimal transformations commute (the constraints are then effectively integrable) or the energy splits into a free and a constrained part. Hamel states the conditions as necessary and sufficient. The sufficiency is correct, and the necessity holds for a single constraint; for several constraints it fails in general, because his argument requires the unsymmetrized coefficient to vanish (Section~\ref{sec:geo-residue}). The conditions contain Hadamard's as the special case of coordinate velocities.
	\item \emph{Korteweg's exception.} For infinitesimal motions about equilibrium---small oscillations---the illegitimate form may be used: the discarded terms are of higher order \cite{korteweg:1900,hamel:1904a}. St\"ackel's survey records this as the one case in which the transformed expression $T^{*}$ always leads to correct results, and notes, more generally, that the illegitimate equations are in general (``im allgemeinen'') incompatible with the legitimate ones, while exceptionally (``ausnahmsweise'') a single equation may come out right \cite[No.~31]{stackel:1908}.
	\item \emph{Hamel's rule for the frame form.} In the frame form one must not set the constrained quasi-velocities to zero in the frame kinetic energy before the momenta are computed (``das gef\"ahrliche $T^{+}$,'' the dangerous $T^{+}$), but terms \emph{quadratic} in them may be dropped in advance \cite[No.~237]{hamel:1949}. Equivalently, only the terms \emph{linear} in the constrained velocities---the constrained momenta \eqref{eq:schneide-Pw}---are dynamically active. Papastavridis states the corresponding conditions in general quasi-velocity form \cite[\S3.5, p.~423]{papastavridis:2002}.
\end{enumerate}

\subsection{The record 1892--2015}\label{sec:ill-record}

The inadmissibility of the substitution became known through incorrect results in the theory of rolling bodies (Table~\ref{tab:classical}). Vierkandt \cite{vierkandt:1892a,vierkandt:1892b} warns against deriving Lagrange's equations from a kinetic energy simplified by the rolling conditions and locates the error in the derivation of the equations (Section~\ref{sec:ill-conditions}, item~1). Korteweg \cite{korteweg:1900} shows that a ``quite common incorrect approach'' to rolling problems rests on this substitution, naming several authors and \S~452 of the first edition of Appell's \emph{Trait\'e}. Appell corrected the treatise and pointed out the error in a separate note \cite[No.~31, fn.~447]{stackel:1908}. His booklet of 1899 states the rule: the Lagrange equations ``cannot be applied when one takes those exceptional constraints into account in order to modify the expression for the \emph{vis viva} $T$,'' a difficulty ``pointed out and studied'' by Neumann, Vierkandt, Hadamard, Carvallo and Korteweg \cite[No.~20]{appell:1899}.\footnote{Quoted from Delphenich's translation; the French text of the booklet was not available to the author.}

Hamel \cite[\S8]{hamel:1904a} asks when Lagrange's equations and the ``illegitimate form'' of the \emph{vis viva} may be used for nonholonomic constraints, credits Neumann with the term, generalizes Hadamard's conditions to quasi-velocities and derives Korteweg's exception for infinitesimal motions. Boltzmann, whose paper of 1902 gives corrected equations for nonholonomic coordinates \cite{boltzmann:1902}, carries the matter into his lectures of 1904. They note that Lagrange's equations require essential modification for nonholonomic coordinates, compute the additional terms and show by a worked example that the ordinary equations ``would yield incorrect equations of motion'' \cite[\S4, p.~16, fn.; \S\S27--28]{boltzmann:1904}.\footnote{``[Die Lagrangeschen Gleichungen in ihrer gew\"ohnlichen Form aber] w\"urden falsche Bewegungsgleichungen liefern''; ``eine Ungereimtheit'' \cite[\S28, p.~115]{boltzmann:1904}.} The encyclopedia surveys of Voss \cite[No.~38]{voss:1901} and St\"ackel \cite[No.~31]{stackel:1908} codify the errors, the terminology, the corrections and the exceptions. A modern account of the period is given in \cite{borisov:2016}.

Between Hamel's textbook and the modern debate the result remained in print (Table~\ref{tab:record}). Neimark and Fufaev \cite{neimark:1972} give the theory of linear nonholonomic systems together with its history. They name Lindel\"of's error without describing it (pp.~1, 64) and mark the substituted function by an asterisk (p.~106). Chaplygin's equations appear as Lagrange's operator on $T^{*}$ plus a correction sum built from the unconstrained momenta; they reduce to Lagrange's form when the constraints are integrable (p.~108). Section~11 states the condition for a single equation to take Lagrange's form with $L^{*}$ (pp.~191--192).

Rosenberg's Chapter~14, ``Embedding Constraints,'' opens with a section ``A Fallacy'': embedding the velocity constraint into $T$ (his $T^{\dagger}$) may give wrong results. He shows this on the nonholonomic constraint $\dot y=z\dot x$ and contrasts it with the holonomic $\dot y=k\dot x$, for which the same operation succeeds; the correct embedding acts on the virtual displacements, never on $T$ \cite[Ch.~14]{rosenberg:1977}. Ardema \cite[\S6.6]{ardema:2005} restates the rule: substitution into $T$ ``in a holonomic system, and only in a holonomic system.''

Flannery \cite{flannery:2005} shows that the variational form of the substitution, the multiplier-augmented Lagrangian inserted into Hamilton's principle, was still being adopted in the 2000s, identifies why it fails and reaffirms the d'Alembertian route. Papastavridis \cite[\S3.5, p.~423]{papastavridis:2002} gives the admissibility conditions in quasi-velocity form.

In geometric mechanics the substituted function is standard under the name \emph{constrained Lagrangian}. Bloch, Krishnaprasad, Marsden and Murray define it by substituting the constraints into $L$ and attach the curvature of the constraint connection to its equations of motion from the outset \cite[\S2.1]{bloch:1996}. Bloch's monograph states it intrinsically, $L_c(q,\dot q)=L(q,\mathrm{hor}\,\dot q)$, with the curvature correction $\delta L_c=\langle \mathbb FL,\,B(\dot q,\delta q)\rangle$ for horizontal variations \cite[\S5.2, Thm.~5.2.2]{bloch:2015}, and Bloch, Marsden and Zenkov state the doctrine compactly \cite{bloch:2005}.

\begin{table}[!htbp]
	\caption{Classical treatments of the substitution question, 1892--1908 (all sources examined at first hand; see the note on sources and translations).}\label{tab:classical}
	\footnotesize\linespread{1}\selectfont
	\renewcommand{\arraystretch}{1.2}\setlength{\tabcolsep}{8pt}
	\begin{tabular}{@{}>{\raggedright\arraybackslash}p{0.235\textwidth}>{\raggedright\arraybackslash}p{0.205\textwidth}>{\raggedright\arraybackslash}p{0.475\textwidth}@{}}
		\toprule
		Source                                    & Terminology                                          & Treatment                                                                                                                                                                                                                         \\
		\midrule
		Vierkandt 1892 \cite{vierkandt:1892a}     & ``im allgemeinen unstatthaft''                       & warning against Lagrange's equations formed from a kinetic energy simplified by the rolling conditions; the error located in the derivation of the equations (\S4)                                                                \\
		Chaplygin 1895/1897 \cite{chaplygin:2002} & forms (7) and (8)                                    & Lindel\"of's error exposed; equations of motion as Lagrange's operator on the substituted function plus correction terms, reducing to Lagrange's form only in the integrable case (pp.~120--121)                                  \\
		Hadamard 1895 \cite{hadamard:1895}        & ---                                                  & conditions on the constraint combinations substitutable for every kinetic energy (Nos.~5--8); all substitutable iff integrable (No.~8)                                                                                            \\
		Appell 1899 \cite{appell:1899}            & ---                                                  & rule stated (No.~20); criterion for a given system, the deviation from Lagrange's form as a quadratic form in the velocities (No.~23, eqs.~(45)--(48)); \emph{Trait\'e} corrected                                                 \\
		Neumann 1899 \cite{neumann:1899}          & legitime/illegitime Form der lebendigen Kraft        & term introduced (p.~437); sign reversal under incomplete substitution (p.~436); rule: set up Lagrange's equations ignoring the constraints and take them into account afterwards                                                  \\
		Korteweg 1900 \cite{korteweg:1900}        & ``ziemlich verbreitete unrichtige Behandlungsweise'' & the substitution identified as the common error in rolling problems, with a list of authors (fn.~2); exception for small oscillations                                                                                             \\
		Voss 1901 \cite{voss:1901}                & illegitime Form (after Neumann)                      & encyclopedia codification: the correction term to Lagrange's form, the term, the sources (No.~38)                                                                                                                                 \\
		Hamel 1904 \cite{hamel:1904a}             & illegitime Form (credited to Neumann)                & conditions for the validity of Lagrange's equations of the retained coordinates, in quasi-velocities (\S8); Korteweg's exception derived                                                                                          \\
		Boltzmann 1904 \cite{boltzmann:1904}      & ---                                                  & lectures: Lagrange's equations require modification for nonholonomic coordinates (\S4); additional terms computed (\S27); worked example with ``falsche Bewegungsgleichungen'' (\S28)                                             \\
		St\"ackel 1908 \cite{stackel:1908}        & illegitime/legitime Gleichungen; $T^{*}$             & encyclopedia codification: the rolling-body errors, Appell's correction, Chaplygin and Korteweg; small oscillations as the one always-safe case; illegitimate equations in general incompatible with the legitimate ones (No.~31) \\
		\botrule
	\end{tabular}
\end{table}

\begin{table}[!htbp]
	\caption{Explicit treatments of the substitution question between 1949 and 2015 (sources examined at first hand).}\label{tab:record}
	\footnotesize\linespread{1}\selectfont
	\renewcommand{\arraystretch}{1.2}\setlength{\tabcolsep}{8pt}
	\begin{tabular}{@{}>{\raggedright\arraybackslash}p{0.235\textwidth}>{\raggedright\arraybackslash}p{0.205\textwidth}>{\raggedright\arraybackslash}p{0.475\textwidth}@{}}
		\toprule
		Source                                                       & Terminology                                     & Treatment                                                                                                                                                                                                                                                                                          \\
		\midrule
		Hamel 1949 \cite{hamel:1949}                                 & illegitime Form; das gef\"ahrliche $T^{+}$      & explicit prohibition in the knife-edge example; separate warning (No.~237); worked counterexample with the term (No.~240)                                                                                                                                                                          \\
		Neimark--Fufaev 1972 \cite{neimark:1972}                     & asterisk convention: $T^{*}$, $L^{*}$           & Chaplygin's equations as Lagrange's operator on $T^{*}$ plus correction terms, vanishing identically for every kinetic energy iff the constraints are integrable (Ch.~III, \S3); condition for a single Lagrange equation with $L^{*}$ (\S11); Lindel\"of's error named, not described (pp.~1, 64) \\
		Rosenberg 1977 \cite{rosenberg:1977}                         & a fallacy; $T^{\dagger}$                        & counterexample pair (nonholonomic vs.\ holonomic); correct embedding via virtual displacements                                                                                                                                                                                                     \\
		Bloch--Krishnaprasad--Marsden--Murray 1996 \cite{bloch:1996} & constrained Lagrangian $L_c$                    & $L_c$ defined by substituting the constraints into $L$; equations of motion with the curvature term of the constraint connection attached (eqs.~(2.1.6)--(2.1.7)); substitution without the curvature term correct in the integrable case (\S2.1)                                                  \\
		Papastavridis 2002 \cite{papastavridis:2002}                 & constrained kinetic energy                      & admissibility conditions in quasi-velocities (p.~423); caution against the claim, found in a multibody text of 1992, that independent coordinates imply arbitrary virtual displacements (p.~419)                                                                                                   \\
		Bloch--Marsden--Zenkov 2005 \cite{bloch:2005}                & constrained Lagrangian                          & substitution legitimate only together with correction terms                                                                                                                                                                                                                                        \\
		Ardema 2005 \cite{ardema:2005}                               & embedding                                       & substitution into $T$ only for holonomic systems; otherwise embedding via virtual displacements                                                                                                                                                                                                    \\
		Flannery 2005 \cite{flannery:2005}                           & adjoined vs.\ embedded                          & variational counterpart: the multiplier-augmented Lagrangian must not be substituted into Hamilton's principle; d'Alembertian route reaffirmed                                                                                                                                                     \\
		Bloch et al.\ 2003/2015 \cite{bloch:2015}                    & constrained Lagrangian $L_c=L\circ\mathrm{hor}$ & substitution always paired with the curvature correction (Thm.~5.2.2); classical literature cited; admissibility question not posed                                                                                                                                                                \\
		\botrule
	\end{tabular}
\end{table}

\section{The modern debate measured against the record}\label{sec:ill-rediscovery}

\subsection{The claims of 2010}\label{sec:deb-claims}

Udwadia and Wanichanon \cite{udwadia:2010} present the knife-edge computation of Section~\ref{sec:ill-schneide} as ``a paradox posed by Hamel in his 1949 book'' and resolve it within the fundamental equation of constrained motion \cite{udwadia:1996,udwadia:2002}. Their bibliography comprises six entries---Hamel's textbook, Gauss (1829), Pars's treatise and three works of the first author's school---and none of the other literature of Sections~\ref{sec:ill-operation}--\ref{sec:ill-record}. Judging by the sources they cite, they were not aware of the historical findings.

The authors state that, while ``Hamel is correct in what he writes,'' he ``leaves the reader unclear as to why'' the substitution fails and ``leaves open the question of identifying those special circumstances'' in which it would fail \cite[p.~1255]{udwadia:2010}. Both statements are incorrect. The ``why'' is stated a few pages after the cited passage in the same chapter, ``vorzeitiger Gebrauch der Bedingungsgleichung'' (premature use of the constraint equation), together with Neumann's term \cite[No.~240]{hamel:1949}. The ``special circumstances'' are the subject of Hadamard's article of 1895 \cite[Nos.~5--14]{hadamard:1895}, of Appell's criterion of 1899 \cite[No.~23]{appell:1899} and of Hamel's \S8 of 1904 (Section~\ref{sec:ill-conditions}). One point must be conceded: the cross-reference Hamel appends at the cited page is defective (Section~\ref{sec:ill-schneide}), and a reader who follows it finds nothing on the subject.

The resolution the paper offers, that the unconstrained system must be conceptualized without the constraints, which are imposed afterwards, is correct: it is Neumann's rule of 1899, Hadamard's premise of 1895 \cite[No.~5]{hadamard:1895} and Appell's statement of the same year \cite[No.~20]{appell:1899}. The mathematical framework in which the constraints are then imposed, the fundamental equation of constrained motion \cite{udwadia:1996,udwadia:2002}, is not in question here; the name and the historical claims are.

\subsection{Three inferences}\label{sec:deb-inferences}

Three further inferences that the paper draws from the example must be examined separately, since they are what makes the example appear to touch the foundations. First, the paper presents the rule as a change of viewpoint, ``from the Lagrangian view of mechanics to the one developed by Gauss'' \cite[p.~1255]{udwadia:2010}, and derives its equation of motion from Gauss's principle of least constraint \cite{gauss:1829,udwadia:1992,udwadia:1996}: with $a:=M^{-1}F$ the acceleration of the unconstrained system, the actual acceleration minimizes the constraint $\tfrac12(\ddot q-a)^{\mathsf T}M(\ddot q-a)$ among all $\ddot q$ with $A\ddot q=b$. For ideal constraints this is d'Alembert's principle in another form, since a convex quadratic function on an affine subspace is minimized exactly where its gradient annihilates the directions of the subspace:
\begin{equation*}
	\ddot q=\operatorname*{arg\,min}_{A\ddot q=b}\ \tfrac12(\ddot q-a)^{\mathsf T}M(\ddot q-a)
	\quad\Longleftrightarrow\quad
	(M\ddot q-F)\cdot\delta q=0\ \text{ for all }\delta q\text{ with }A\,\delta q=0 .
\end{equation*}

Gauss's principle turns the pointwise principle of d'Alembert into a minimum principle at each instant \cite[Ch.~IV, \S8, p.~106]{lanczos:1962}, \cite[\S\S6.2, 6.4]{papastavridis:2002}, in full geometric generality \cite[Thm.~7.1]{lewis:1996}. The authors of the framework themselves call the principles of d'Alembert, Lagrange and Gauss ``all-encompassing'' and exclude a new fundamental principle \cite{udwadia:1992}, and the 2010 paper notes that ``both views ultimately rest on d'Alembert's principle'' \cite[p.~1255]{udwadia:2010}.

The change of viewpoint therefore cannot bear on the question. Both principles take the same input---the mass matrix and the forces of the unconstrained system---and it is this input that the substitution falsifies. The three-step conceptualization of the 2010 paper, with the unconstrained system set up first and the constraints imposed afterwards, is Neumann's rule of 1899 restated in Gauss's framework. The ``conflation'' of the unconstrained with the constrained system is the illegitimate form under another name. What the paper does not supply is an analysis of the substitution itself: which term it discards, and under which conditions $T^{+}$ may nevertheless be used. The classical conditions of Section~\ref{sec:ill-conditions} and the residue \eqref{eq:residue} answer that question. The relative merits of the two principles as working tools lie outside the present paper.

Second, the paper states that its equation works ``in situations [\dots] in which the Lagrange multiplier method breaks down'' and obviates ``the difficulties in finding the Lagrange multipliers'' \cite[p.~1256]{udwadia:2010}; the situations named are functionally dependent and nonlinear constraints. Neither occurs in the example, whose single constraint is linear and independent. The error occurs in forming $T$, before any multiplier is introduced, and the multiplier route is the one by which Hamel---and the paper itself in its eqs.~(1.9)--(1.12)---obtains the correct equations. For consistent ideal constraints $A\ddot q=b$ the fundamental equation is the multiplier solution in closed form,
\begin{equation*}
	M\ddot q=F+A^{\mathsf T}\lambda ,
	\qquad
	\lambda=(AM^{-1}A^{\mathsf T})^{+}\bigl(b-AM^{-1}F\bigr),
\end{equation*}
where the pseudoinverse selects one multiplier vector when the constraints are dependent; its derivation from the Lagrange and Maggi equations is given by Zegzhda et al.\ \cite{zegzhda:2016}. Whatever the merits of the closed form as a computational tool, the example does not exhibit them.

Third, the paper concludes that a rigorous explanation leads ``deeper to the foundations of analytical dynamics and to the use of newly developed concepts that deal with singular mass matrices'' \cite[p.~1265]{udwadia:2010}. Singular mass matrices arise only when $T^{+}$ is itself treated as the kinetic energy of an unconstrained system in the original coordinates. That the singular-mass equation then reproduces the wrong equations shows the consistency of the procedure, not the source of the error. The explanation of the example is the residue \eqref{eq:residue}, and it was complete in 1895--1904.

\subsection{The subsequent papers}\label{sec:deb-subsequent}

The subsequent literature adopted the name from this single source. Chen \cite{chen:2013} states that the ``intricacy surrounding the Hamel's embedding method'' is called the Hamel paradox as ``first phrased in'' the 2010 paper, ``the only known work regarding the Hamel paradox.'' He formalizes Rosenberg's observations into the \emph{Rosenberg conjecture} (substitution legitimate exactly for holonomic constraints), confirms the nonholonomic half and subjects the holonomic half to a criterion that his own holonomic examples pass (Appendix~\ref{app:chen}). Wanichanon and Cho \cite{wanichanon:2024}, reading the criterion as the claim that holonomic embedding can fail, reply that the \emph{complete} embedding of a holonomic constraint yields correct equations. Neumann's example of 1899 \cite[p.~436]{neumann:1899} contains the entire content of this exchange---an integrable constraint, an incomplete substitution, wrong equations, and the resolution---and Hadamard's theorem of 1895 \cite[Nos.~5, 8]{hadamard:1895} is, for velocity-level substitution, the Rosenberg conjecture. The circumstances under which the record was not consulted are discussed in Section~\ref{sec:discussion}.

\section{Systematic resolution}\label{sec:resolution}

This section restates the findings of Section~\ref{sec:illegitimate} in geometric terms: a single construction, the restriction of the mass metric to the constraint distribution, accounts for the failure, contains the classical admissibility conditions and reconciles the classical prohibition with the constrained-Lagrangian practice of geometric mechanics.

\subsection{Setting}\label{sec:geo-setting}

The configuration space is an $n$-dimensional manifold $Q$ with local coordinates $q^i$ ($i,j=1,\dots,n$); the coordinate vector fields $\partial_i:=\partial/\partial q^i$ and the coordinate one-forms $\mathrm dq^i$ form the dual bases of $T_qQ$ and $T_q^{*}Q$, $\langle\mathrm dq^i,\partial_j\rangle=\delta^i_j$, and a vector $w\in T_qQ$ with components $w^i$ is $w=w^i\partial_i$. The kinetic energy defines the \emph{mass metric} $\mathcal M$,
\begin{equation}\label{eq:massmetric}
	T=\tfrac12\,\mathcal M_{ij}(q)\,\dot q^i\dot q^j ;
\end{equation}
we treat the scleronomic catastatic case ($B_r=0$ in \eqref{eq:pfaff}) and comment on time-dependent data where relevant. With $m$ constraints, the number of degrees of freedom is $f:=n-m$; the index range $f+1,\dots,n$ is reserved for the constraint directions, in anticipation of the adapted frames of Section~\ref{sec:geo-frames}. The coefficient rows of \eqref{eq:pfaff-virtuell} are the components of the \emph{constraint one-forms}
\begin{equation}\label{eq:constraintforms}
	\theta^{f+r}:=B_{ri}(q)\,\mathrm dq^i ,
	\qquad r=1,\dots,m ,
\end{equation}
and the virtual conditions state, invariantly, that these forms annihilate the virtual displacements. They thus define at each configuration the subspace
\begin{equation}\label{eq:distribution}
	D_q=\bigl\{\,w\in T_qQ:\ \bigl\langle\theta^{f+r}_q,\,w\bigr\rangle=0,\ r=1,\dots,m\,\bigr\},
	\qquad \dim D_q=f ,
\end{equation}
and the assignment $q\mapsto D_q$ is the \emph{constraint distribution} $D\subset TQ$: admissible velocities and virtual displacements are exactly its elements, and the $\theta^{f+r}$ span the annihilator of $D$.

In this language the d'Alembert--Lagrange principle \eqref{eq:dalembert} is one covector statement: along the motion, the \emph{balance one-form}
\begin{equation}\label{eq:balance}
	R:=\bigl(\EL{i}(T)-Q_i\bigr)\,\mathrm dq^i
\end{equation}
annihilates $D$, that is, $\langle R,w\rangle=0$ for every $w\in D$. Route A of Section~\ref{sec:ill-operation} comprises the three ways of resolving this statement: expanding $R$ in the annihilator basis gives the multiplier form ($R=\lambda_r\,\theta^{f+r}$); pairing $R$ with a basis $\{w_\alpha\}_{\alpha=1}^{f}$ of $D$ gives Maggi's equations,
\begin{equation}\label{eq:maggi}
	\bigl\langle R,\,w_\alpha\bigr\rangle=0 ,
	\qquad \alpha=1,\dots,f ;
\end{equation}
and resolving the balance in an anholonomic frame adapted to $D$ gives the Boltzmann--Hamel equations, to which we now turn.

\subsection{Anholonomic frames and the frame form of the equations}\label{sec:geo-frames}

Let $\{e_\lambda\}$, $\lambda,\mu,\rho,\sigma=1,\dots,n$, be a frame field, $e_\lambda=b^i_\lambda\,\partial_i$, with dual coframe $\theta^\lambda=a^\lambda_i\,\mathrm dq^i$, $a^\lambda_i b^i_\mu=\delta^\lambda_\mu$. Velocities and virtual displacements are resolved as
\begin{equation}\label{eq:quasi}
	\omega^\lambda:=a^\lambda_i\,\dot q^i ,
	\qquad
	\delta\pi^\lambda:=a^\lambda_i\,\delta q^i
\end{equation}
(quasi-velocities and quasi-displacements); invariantly, $\omega^\lambda$ and $\delta\pi^\lambda$ are the natural pairings of the coframe element $\theta^\lambda$ with the velocity and with the virtual displacement, respectively. The \emph{structure functions} of the frame are defined by
\begin{equation}\label{eq:cdef}
	[e_\sigma,e_\rho]=c^\lambda_{\sigma\rho}\,e_\lambda ,
	\qquad
	c^\lambda_{\sigma\rho}
	=b^i_\sigma b^j_\rho\bigl(\partial_j a^\lambda_i-\partial_i a^\lambda_j\bigr),
\end{equation}
and vanish identically exactly when the frame is a coordinate frame. The name is Cartan's: the $c^\lambda_{\sigma\rho}$ measure the failure of the frame fields to commute, and they are equivalently the coefficients of the structure equations of the coframe, $\mathrm d\theta^\lambda=-\tfrac12\,c^\lambda_{\sigma\rho}\,\theta^\sigma\wedge\theta^\rho$; for a left-invariant frame on a Lie group they reduce to the structure \emph{constants} of its Lie algebra, and the structure equations become the Maurer--Cartan equations \cite[\S15.3a, eqs.~(15.23)--(15.24)]{frankel:2012}. In analytical mechanics they entered through Hamel's paper of 1904 and are known there as the \emph{Hamel coefficients} (or Boltzmann--Hamel symbols) \cite{hamel:1904a,papastavridis:2002,muller:2021}; the shift of name from ``constants'' to ``functions'' records that for a general frame on a configuration space they depend on the configuration. The theory rests on two short computations.

\begin{lemma}[noncommutation; transpositional relation]\label{lem:noncomm}
	For variations satisfying $\mathrm d\delta q^i=\delta\mathrm dq^i$,\footnote{This is a hypothesis on the variations, not an identity. It is needed only for the variational reading of \eqref{eq:central}; the frame identity of Lemma~\ref{lem:frameid}, and with it the Boltzmann--Hamel equations, follow from the chain rule alone (Appendix~\ref{app:frameid}). Whether the hypothesis is compatible with Chetaev-admissible variations in the presence of velocity constraints is a separate, variational question; see the end of Section~\ref{sec:geo-residue}.}
	\begin{equation}\label{eq:noncomm}
		\delta\omega^\lambda-\frac{\mathrm d}{\mathrm dt}\,\delta\pi^\lambda
		= c^\lambda_{\sigma\rho}\,\omega^\sigma\,\delta\pi^\rho .
	\end{equation}
\end{lemma}

\begin{proof}
	$\frac{\mathrm d}{\mathrm dt}\delta\pi^\lambda-\delta\omega^\lambda
		=\dot a^\lambda_i\,\delta q^i-(\delta a^\lambda_i)\,\dot q^i
		=\bigl(\partial_j a^\lambda_i-\partial_i a^\lambda_j\bigr)\dot q^j\delta q^i
		=-\,c^\lambda_{\sigma\rho}\,\omega^\sigma\delta\pi^\rho$,
	by \eqref{eq:quasi}, \eqref{eq:cdef}, after the terms with $\mathrm d\delta q=\delta\mathrm dq$ cancel.
\end{proof}

Equation \eqref{eq:noncomm} is Hamel's transpositional relation \cite{hamel:1904a,hamel:1904b}; for a coordinate frame its right-hand side vanishes. Modern treatments are given in \cite{papastavridis:2002,flannery:2011}, and the admissible hypotheses are still being examined \cite{talamucci:2025,talamucci:2026}. Vierkandt's localization of the failure (Section~\ref{sec:ill-conditions}, item~1) is the observation that the derivation of Lagrange's equations relies on the vanishing of this deviation, which anholonomic velocity combinations do not provide.

Write $\overline T(q,\omega):=T$ for the kinetic energy as a function of the quasi-velocities, $\overline{\mathcal M}_{\lambda\mu}:=\mathcal M(e_\lambda,e_\mu)$ for the frame components of the mass metric, so that $\overline T=\tfrac12\overline{\mathcal M}_{\lambda\mu}\omega^\lambda\omega^\mu$, and
\begin{equation}\label{eq:framemomenta}
	p_\lambda:=\frac{\partial\overline T}{\partial\omega^\lambda}
	=\overline{\mathcal M}_{\lambda\mu}\,\omega^\mu ,
	\qquad
	e_\sigma(\overline T):=b^i_\sigma\,
	\Bigl(\frac{\partial\overline T}{\partial q^i}\Bigr)_{\!\omega}
\end{equation}
for the frame momenta and the frame derivative at frozen $\omega$.

Two covectors organize the bookkeeping: the \emph{momentum covector} and the \emph{applied-force covector},
\begin{equation}\label{eq:momentumcovector}
	p:=\frac{\partial T}{\partial\dot q^i}\,\mathrm dq^i=p_\lambda\,\theta^\lambda ,
	\qquad
	F:=Q_i\,\mathrm dq^i ,
\end{equation}
so that $\langle p,w\rangle=p_\lambda\,\delta\pi^\lambda$ for a virtual displacement field $w=\delta q^i\,\partial_i=\delta\pi^\lambda e_\lambda$, independently of the frame. For variations satisfying $\mathrm d\delta q^i=\delta\,\mathrm dq^i$, the definitions alone give the identity
\begin{equation}\label{eq:central}
	\frac{\mathrm d}{\mathrm dt}\bigl\langle p,\,w\bigr\rangle-\delta T
	=\bigl\langle R,\,w\bigr\rangle+\bigl\langle F,\,w\bigr\rangle ,
\end{equation}
and along the motion, for admissible $w$ ($\langle R,w\rangle=0$), the \emph{central equation} of Heun and Hamel,
\begin{equation}\label{eq:central-motion}
	\frac{\mathrm d}{\mathrm dt}\bigl\langle p,\,w\bigr\rangle-\delta T
	=\bigl\langle F,\,w\bigr\rangle
	\qquad(w\in D)
\end{equation}
\cite{hamel:1904a,papastavridis:2002}. Expanding $\langle p,w\rangle=p_\lambda\delta\pi^\lambda$ and $\delta\overline T=e_\lambda(\overline T)\,\delta\pi^\lambda+p_\lambda\,\delta\omega^\lambda$ in \eqref{eq:central} and eliminating $\delta\omega^\lambda$ by Lemma~\ref{lem:noncomm} yields, coefficient by coefficient, the following frame identity; the transpositional deviation $p_\lambda\bigl(\delta\omega^\lambda-\frac{\mathrm d}{\mathrm dt}\delta\pi^\lambda\bigr)$ is precisely the structure term. (A direct componentwise proof is given in Appendix~\ref{app:frameid}.)

\begin{lemma}[frame identity]\label{lem:frameid}
	For any frame and any $T$ of the form \eqref{eq:massmetric},
	\begin{equation}\label{eq:frameid}
		b^i_\sigma\,\EL{i}(T)
		=\frac{\mathrm d}{\mathrm dt}\,p_\sigma
		- e_\sigma(\overline T)
		+ c^\lambda_{\sigma\rho}\,\omega^\rho\,p_\lambda .
	\end{equation}
\end{lemma}

The identity holds for every motion and every frame; no hypothesis on variations enters it. The step that follows---contracting the balance with the admissible frame fields---uses only the pointwise principle \eqref{eq:dalembert}. Contracting the balance \eqref{eq:balance} with the frame therefore resolves the d'Alembert--Lagrange principle into
\begin{equation}\label{eq:bh}
	\frac{\mathrm d}{\mathrm dt}\frac{\partial\overline T}{\partial\omega^\sigma}
	- e_\sigma(\overline T)
	+ c^\lambda_{\sigma\rho}\,\omega^\rho\,\frac{\partial\overline T}{\partial\omega^\lambda}
	=\overline Q_\sigma ,
	\qquad
	\overline Q_\sigma:=b^i_\sigma Q_i ,
\end{equation}
the \emph{Boltzmann--Hamel equations} \cite{boltzmann:1902,hamel:1904a}. For a coordinate frame ($c\equiv0$) they are Lagrange's equations; for the body frame of rigid-body kinematics ($c^\lambda_{\sigma\rho}=\varepsilon_{\sigma\rho\lambda}$) they reduce to Euler's equations---the standard consistency checks.

\paragraph{Adapted frames and correct embedding.}
Now adapt the frame to the constraints: take $e_\alpha$, $\alpha,\beta,\gamma=1,\dots,f$, to be a basis of $D$, and let the last $m$ coframe elements be the constraint forms, $\theta^s$, $s=f{+}1,\dots,n$, proportional to the rows $B_{ri}$. Then the constraints read $\omega^s=0$, the virtual conditions read $\delta\pi^s=0$, and the principle retains exactly the admissible components of \eqref{eq:bh}:
\begin{equation}\label{eq:bh-embedded}
	\begin{gathered}
		\boxed{\ \frac{\mathrm d}{\mathrm dt}\frac{\partial\overline T}{\partial\omega^\alpha}
			- e_\alpha(\overline T)
			+ c^\lambda_{\alpha\rho}\,\omega^\rho\,\frac{\partial\overline T}{\partial\omega^\lambda}
			=\overline Q_\alpha\ }\\[4pt]
		\text{with } \omega^s=0 \text{ imposed \emph{after} the differentiations}
	\end{gathered}
\end{equation}

These $f$ equations, together with the $m$ constraints, are the correctly embedded equations of motion---multiplier-free, and equivalent to Maggi's form by construction. Note that the sum over $\lambda$ in \eqref{eq:bh-embedded} includes the constrained directions: their contribution enters through the \emph{constrained momenta}
\begin{equation}\label{eq:Ps}
	P_s:=p_s\big|_{\omega^s=0}=\overline{\mathcal M}_{s\beta}\,\omega^\beta ,
\end{equation}
the general form of the quantity \eqref{eq:schneide-Pw} isolated in the knife-edge example.

\subsection{Restriction of the mass metric; the integrable case}\label{sec:geo-diagnosis}

In this setting the illegitimate form has a simple description. Substituting the constraints into the kinetic energy evaluates $\overline T$ at $\omega^s=0$:
\begin{equation}\label{eq:Tplus}
	\overline T^{+}(q,\omega^\alpha)
	:=\overline T\big|_{\omega^s=0}
	=\tfrac12\,\overline{\mathcal M}_{\alpha\beta}\,\omega^\alpha\omega^\beta ,
\end{equation}
which is the quadratic form of the mass metric \emph{restricted to the distribution} $D$; the coordinate function $T^{*}$ of \eqref{eq:Tstar} is the same object read in a coordinate description. The restriction is always well defined: it is the kinetic energy of the admissible states, and the constrained motion has exactly this energy. Whether it is the induced energy of a configuration manifold, on which Lagrange's equations apply, is decided by the Frobenius integrability of $D$.

\begin{proposition}[integrable case; complete versus incomplete embedding]\label{prop:leaf}
	Let $D$ be integrable.
	(a) In leaf-adapted coordinates ($D=\operatorname{span}\{\partial_\alpha\}$, constraints $\dot z^s=0$), the adapted frame is a coordinate frame, the structure functions vanish, and \eqref{eq:bh-embedded} reduces to the Lagrange equations of $\overline T^{+}$ on the leaf: the restricted metric is the induced metric of the integral manifold, and Route B, executed as this \emph{complete} embedding (dependent coordinates eliminated together with their velocities), coincides with Route A.
	(b) If instead only the velocities are eliminated while the corresponding coordinates survive as arguments of $T^{*}$, varied or held fixed, the resulting equations are in general incorrect.
\end{proposition}

\begin{proof}
	(a) is \eqref{eq:bh-embedded} with $c\equiv0$. For (b) a counterexample suffices. The oldest is Neumann's \cite[p.~436]{neumann:1899}: a point mass under the velocity constraint $\dot x=\varkappa x$, integrable in the rheonomic sense, for which the Lagrange equation of the eliminated coordinate $x$, formed from the substituted energy, has the sign of its inertia term reversed. For the equations of the retained coordinates, the integrable constraint $\dot y=y\,\dot x$ on a free particle serves as an example: the Lagrange equation of $T^{*}$ in $x$ with $y$ held fixed differs from the correct equation by the coordinate-slot term $y\,\partial T^{*}\!/\partial y$ of Proposition~\ref{prop:decomposition}, while the residue vanishes because the adapted frame fields commute (Appendix~\ref{app:examples-pair}).
\end{proof}

Proposition~\ref{prop:leaf} settles the holonomic sub-debate of \cite{chen:2013,wanichanon:2024}: ``complete embedding'' is the passage to the integral manifold, and only that passage is covered by the legitimacy of case (a). For the knife edge no such passage exists: Figure~\ref{fig:twist} shows the admissible planes above a fixed contact point twisting with $\vartheta$; they are the tangent planes of no surface.

\begin{figure}[t]
	\centering
	\includegraphics[scale=0.6]{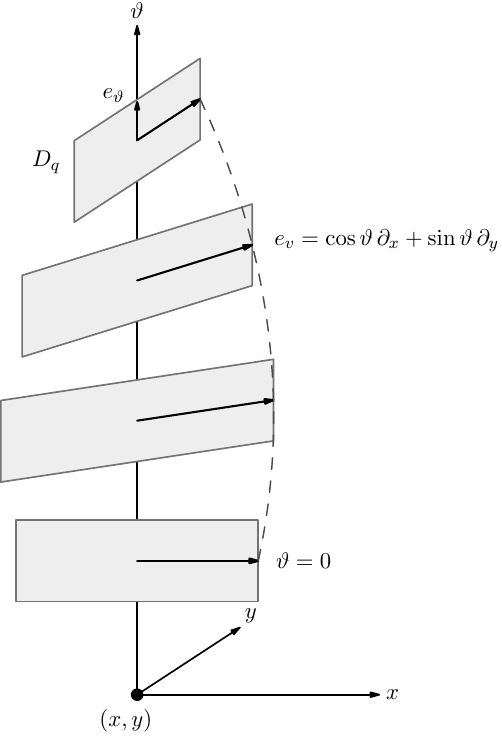}
	\caption{The constraint distribution of the knife edge above a fixed contact point $(x,y)$. At height $\vartheta$ the admissible plane $D_q$ is spanned by the rotation direction $e_\vartheta$ and the blade direction $e_v=\cos\vartheta\,\partial_x+\sin\vartheta\,\partial_y$, which rotates with $\vartheta$. The planes twist like a spiral staircase and are the tangent planes of no surface: $D$ is not integrable (Frobenius's condition fails: $c^{w}_{v\vartheta}\neq0$, Section~\ref{sec:geo-residue}, item~4), no leaf exists, and the restricted energy is the induced energy of no configuration submanifold (Proposition~\ref{prop:leaf}).}\label{fig:twist}
\end{figure}

\subsection{The residue and the criterion}\label{sec:geo-residue}

\begin{proposition}[residue of the illegitimate reduction]\label{prop:residue}
	Fix a constraint-adapted frame (scleronomic, catastatic). Along the constrained motion ($\omega^s=0$), the correctly embedded equations \eqref{eq:bh-embedded} differ from the frame equations formed from the restricted energy $\overline T^{+}$ alone,
	\begin{equation}\label{eq:reduced-op}
		\frac{\mathrm d}{\mathrm dt}\frac{\partial\overline T^{+}}{\partial\omega^\alpha}
		- e_\alpha(\overline T^{+})
		+ c^\beta_{\alpha\gamma}\,\omega^\gamma\,\frac{\partial\overline T^{+}}{\partial\omega^\beta}
		=\overline Q_\alpha ,
	\end{equation}
	by exactly the term
	\begin{equation}\label{eq:residue}
		\begin{gathered}
			\boxed{\ \mathrm{Res}_\alpha
				= c^{s}_{\alpha\beta}\,\omega^\beta\,P_s
				= c^{s}_{\alpha\beta}\,\overline{\mathcal M}_{s\gamma}\;\omega^\beta\omega^\gamma\ }\\[4pt]
			(\text{sum over }s>f;\ \beta,\gamma\le f)
		\end{gathered}
	\end{equation}
\end{proposition}

\begin{proof}
	Decompose $\overline T=\overline T^{+}+\overline{\mathcal M}_{s\beta}\,\omega^s\omega^\beta+\tfrac12\overline{\mathcal M}_{ss'}\,\omega^s\omega^{s'}$ by powers of the constrained quasi-velocities; the two mixed blocks of the quadratic form merge by the symmetry of $\overline{\mathcal M}$, which is where the factor $\tfrac12$ of the mixed block disappears. At $\omega^s=0$: the admissible momenta agree, $p_\alpha|_0=\partial\overline T^{+}\!/\partial\omega^\alpha$ (hence so do their time derivatives along constrained motions); the frame derivatives agree, $e_\alpha(\overline T)|_0=e_\alpha(\overline T^{+})$; and the constrained momenta are $p_s|_0=P_s$ of \eqref{eq:Ps}. In the structure sum of \eqref{eq:bh-embedded}, $\omega^\rho|_0$ retains only admissible $\rho=\beta$, and splitting $\lambda$ into admissible and constrained values yields the two groups of \eqref{eq:reduced-op} and \eqref{eq:residue}.
\end{proof}

Figure~\ref{fig:frameplane} shows the geometry of Proposition~\ref{prop:residue} on the knife edge: the constrained momentum $P_w$ is the component of the momentum covector along the coordinate complement, and it is the tilt of the level lines of $\overline T$ that carries it.

\begin{figure}[t]
	\centering
	\includegraphics[scale=0.6]{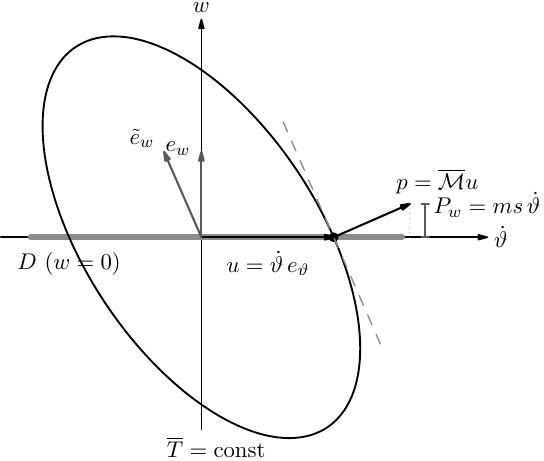}
	\caption{The frame plane $(\dot\vartheta,w)$ of the knife edge at $v=0$. The level line of the frame kinetic energy $\overline T=\tfrac12mw^{2}+ms\dot\vartheta\,w+\tfrac12I_B\dot\vartheta^{2}$ is an ellipse tilted by the mixed term. The distribution $D$ is the axis $w=0$; the coordinate complement $e_w$ is the vertical direction, and the $\mathcal M$-orthogonal complement $\tilde e_w=e_w-(ms/I_B)\,e_\vartheta$ is the diameter conjugate to $D$, the tangent direction of the ellipse where it crosses $D$ (dashed). For a velocity $u=\dot\vartheta\,e_\vartheta$ on $D$ the momentum covector $p=\overline{\mathcal M}u$ has the $w$-component $P_w=ms\dot\vartheta$, the constrained momentum \eqref{eq:Ps} that the restricted energy discards; in the frame with $\tilde e_w$ this component vanishes, and the coupling reappears in the structure functions (Remark~\ref{rem:orthogonal}). Drawn for $m=1$, $I_B=1.6$, $ms=0.7$.}\label{fig:frameplane}
\end{figure}

The residue \eqref{eq:residue} contains the classical conditions:

\begin{enumerate}
	\item \emph{Hamel's rule for the frame form} (No.~237; Section~\ref{sec:ill-conditions}, item~5): terms quadratic in the constrained velocities never contribute---only the linear terms, i.e., the couplings $\overline{\mathcal M}_{s\beta}$, are dynamically active.
	\item \emph{Admissibility.} $\mathrm{Res}\equiv0$ for all motions exactly when the part of the coefficient
	      \begin{equation}\label{eq:Kcoeff}
		      K_{\alpha\beta\gamma}:=c^{s}_{\alpha\beta}\,\overline{\mathcal M}_{s\gamma}
	      \end{equation}
	      symmetric in $\beta\gamma$ vanishes. Two structural mechanisms produce this, direction by direction: $\overline{\mathcal M}_{s\beta}=0$ (the energy splits; no metric coupling between $D$ and $e_s$), or $c^{s}_{\alpha\beta}=0$ on $D$ (the Frobenius tensor of $D$ has no component along $e_s$). They are \emph{sufficient}; for a single constraint they are also necessary at each configuration (Appendix~\ref{app:hamel-par8}); for several constraints they are not, not even after a linear recombination of the constraints. Under a recombination $\tilde\theta^{s}=\Lambda^{s}_{s'}(q)\,\theta^{s'}$ with the complement kept, $K$ is invariant: on $D$ the structure functions transform with $\Lambda$, because $\mathrm d\Lambda^{s}_{s'}\wedge\theta^{s'}$ vanishes on admissible pairs, the couplings transform with $\Lambda^{-1}$, and the two cancel in the sum over $s$. If, after some recombination, every constrained direction satisfied one of the two mechanisms, every term of that sum would vanish and $K$ with it; a cancellation with $K\neq0$ therefore cannot be produced by the two mechanisms in any recombination. Such a $K$ is antisymmetric in $\alpha\beta$; since its symmetric part in $\beta\gamma$ vanishes, it is antisymmetric in $\beta\gamma$ as well, hence totally antisymmetric. This requires $f\ge3$; and since the covectors $K_{\alpha\beta\,\cdot}$ lie in the span of the $m$ coupling rows, it also requires $m\ge3$ (Appendix~\ref{app:examples-cancel}). The same structure appears on the variational side: the compatibility of Chetaev-admissible variations with the transition equation $\mathrm d\delta q=\delta\mathrm dq$ coincides with integrability for a single constraint, while for several constraints the non-integrable parts can compensate one another \cite[Sect.~5.1]{talamucci:2026}.

	      The exact condition for a given system, in coordinates and for a substituted energy whose coefficients depend on the retained coordinates alone, is Appell's criterion of 1899 (Section~\ref{sec:ill-conditions}, item~2). Hamel's \S8 conditions of 1904 concern the Lagrange equations of the retained coordinates with the abbreviated energy free of the eliminated ones. They assert that, after a suitable linear recombination of the constraints, the two mechanisms exhaust the possibilities; his extreme cases are ``the energy splits into two separate parts'' and ``the admissible infinitesimal transformations commute,'' and the latter is integrability, i.e., Hadamard's case \cite[No.~8]{hadamard:1895}. His argument, however, requires the unsymmetrized coefficient $K$ itself to vanish (his equation (8)), which is more than the residue demands. The necessity claim therefore does not hold in this generality, and the example with the coordinate complement lies within the hypotheses of his theorem (Appendix~\ref{app:hamel-par8}).
	\item \emph{Korteweg's exception.} $\mathrm{Res}$ is quadratic in the velocities; in the theory of small oscillations about an equilibrium at rest it is of second order and drops out---Section~\ref{sec:ill-conditions}, item~4. About a steady motion with nonzero velocity its linearization keeps terms of first order.
	\item \emph{The knife edge.} With the frame $(e_v,e_w,e_\vartheta)$ of Section~\ref{sec:ill-schneide} one finds $c^{w}_{v\vartheta}=-1$, $c^{w}_{\vartheta v}=+1$, $P_w=ms\dot\vartheta$, and \eqref{eq:residue} returns $\mathrm{Res}_v=-ms\dot\vartheta^2$ and $\mathrm{Res}_\vartheta=+ms\dot\vartheta v$---exactly the two terms lost in \eqref{eq:schneide-falsch} (Figure~\ref{fig:frameplane}).
\end{enumerate}

\subsection{The coordinate shortcut and Bloch's equations}\label{sec:geo-decomposition}

Proposition~\ref{prop:residue} compares two computations \emph{within the frame calculus}; it quantifies variant (B3). The coordinate-level shortcut (B1) commits a second, independent error, and the two can be separated exactly. For constraints in solved form the coordinates single out a frame: the admissible fields $w_\alpha:=\partial_\alpha+A^{s}_{\alpha}\,\partial_s$ and the \emph{coordinate complement} $e_s:=\partial_s$.

\begin{proposition}[decomposition of the full-elimination shortcut]\label{prop:decomposition}
	Let the catastatic constraints be given in solved form, $\dot q^{s}=A^{s}_{\beta}(q)\,\dot q^{\beta}$, and let $T^{*}(q,\dot q^{\beta})$ be the fully eliminated energy of variant (B1). Then, along the constrained motion,
	\begin{equation}\label{eq:decomposition}
		\boxed{\ \EL{\alpha}\bigl(T^{*}\bigr)-\bigl(Q_\alpha+A^{s}_{\alpha}Q_s\bigr)
			\;=\;
			A^{s}_{\alpha}\,\frac{\partial T^{*}}{\partial q^{s}}
			\;-\;
			c^{s}_{\alpha\beta}\,\dot q^{\beta}P_s\ }
	\end{equation}
	with $c^{s}_{\alpha\beta}=w_\alpha\bigl(A^{s}_{\beta}\bigr)-w_\beta\bigl(A^{s}_{\alpha}\bigr)$ the structure functions of the coordinate-adapted frame (the Frobenius data of the solved form) and $P_s$ the constrained momenta \eqref{eq:Ps} taken in the coordinate complement.
\end{proposition}

\begin{proof}
	For this frame the coframe is $\theta^\alpha=\mathrm dq^\alpha$, $\theta^{s}=\mathrm dq^{s}-A^{s}_{\beta}\mathrm dq^{\beta}$, so $\omega^\alpha=\dot q^\alpha$ and $\overline T^{+}=T^{*}$ \emph{exactly}; moreover $c^{\gamma}_{\alpha\beta}=0$, since $[w_\alpha,w_\beta]$ has only $\partial_s$-components. The frame identity \eqref{eq:frameid}, contracted with $w_\alpha$ and evaluated on $\omega^{s}=0$, therefore reads
	$w^i_\alpha\,\EL{i}(T)
		=\frac{\mathrm d}{\mathrm dt}\frac{\partial T^{*}}{\partial\dot q^{\alpha}}
		-\partial_\alpha T^{*}-A^{s}_{\alpha}\,\partial_s T^{*}
		+c^{s}_{\alpha\beta}\,\dot q^{\beta}P_s$,
	and the d'Alembert--Lagrange principle sets the left side equal to the projected force $Q_\alpha+A^{s}_{\alpha}Q_s$. Rearranging gives \eqref{eq:decomposition}.
\end{proof}

Equation~\eqref{eq:decomposition} is not new in substance. With $Q=0$ it is equation (4) of Bloch, Marsden and Zenkov \cite{bloch:2005}, where $c^{s}_{\alpha\beta}$ appears as the curvature of the Ehresmann connection defined by the constraints; and for Chaplygin systems, in which neither $T$ nor $A^{s}_{\beta}$ depends on the $q^{s}$, the first term on the right vanishes and the equation is Chaplygin's own equation (7) of 1897 \cite[p.~120]{chaplygin:2002}, in the form Neimark and Fufaev transmit it \cite[eq.~(3.17)]{neimark:1972}. What is new here is only the reading: the two terms are the two errors of the coordinate-level shortcut---the coordinate-slot term and the residue of Proposition~\ref{prop:residue}---quantified separately.

The two terms on the right are the two classical failure modes in pure form. The \emph{coordinate-slot term} $A^{s}_{\alpha}\,\partial T^{*}\!/\partial q^{s}$ is active whenever the eliminated coordinates survive as arguments of $T^{*}$; the example $\dot y=y\dot x$ of Proposition~\ref{prop:leaf} shows it alone, and Neumann's example \cite[p.~436]{neumann:1899} shows the same chain-rule mechanism in the equation of an eliminated coordinate. The \emph{residue of the coordinate complement} is the term of Proposition~\ref{prop:residue} in the frame that the coordinates select. Variant (B1) is exact precisely when the right-hand side vanishes along every motion. Hamel's conditions of 1904 are the anatomy of this requirement: his demand that the abbreviated $T$ be free of the eliminated coordinates removes the first term, and his bracket and splitting conditions remove the second (Section~\ref{sec:ill-conditions}, Appendix~\ref{app:hamel-par8}). For the knife edge with $s=0$ and for Rosenberg's fallacy \cite[Ch.~14]{rosenberg:1977} the formula returns the errors of the shortcut term by term (Appendix~\ref{app:examples-pair}).

If $D$ is not integrable, the restricted energy is still well defined: $\overline T^{+}$ is a fiberwise quadratic form on the vector bundle $D$---exactly the constrained Lagrangian of geometric mechanics \cite{bloch:1996,bloch:2005}. What fails is more specific: no leaf exists (Figure~\ref{fig:twist}), so $\overline T^{+}$ is not the induced kinetic energy of any configuration submanifold, and the constrained system is not a Lagrangian system in fewer coordinates. And $\overline T^{+}$ alone does not determine the dynamics: the restricted energy of the knife edge, $\tfrac12mv^{2}+\tfrac12I_B\dot\vartheta^{2}$, is independent of the offset $s$ at fixed $I_B$, while the correct equations \eqref{eq:schneide-korrekt} depend on $s$ explicitly. The discarded information---the constrained momenta $P_s$, i.e., the coupling of $D$ to its complement in the mass metric---is precisely what the correct equations retain, apportioned in one way or another between the momenta and the structure functions (Remark~\ref{rem:orthogonal}).

The derivation above is scleronomic and catastatic; Neumann's affine example \cite[p.~436]{neumann:1899} and Hamel's rheonomic gear \cite[No.~240]{hamel:1949} lie outside it, and both extensions add terms of their own. For affine values $\omega^{s}=\psi^{s}(q)$ the restricted energy depends on the configuration also through the substituted values: its frame derivative differs from $e_\alpha(\overline T)$ on the constrained motion by the chain-rule term $P_s\,e_\alpha(\psi^{s})$, the quadratic block of $\overline T$ contributes to $P_s$, and the structure sum keeps terms with a constrained lower index.

For a time-dependent coframe $a(q,t)$ the proof of Appendix~\ref{app:frameid} goes through with $\mathrm da^\lambda_i/\mathrm dt$ acquiring the explicit term $\partial_t a^\lambda_i$: the frame identity \eqref{eq:frameid} acquires the \emph{transport term} $+\,d^{\lambda}_{\sigma}\,p_\lambda$ with $d^{\lambda}_{\sigma}:=b^i_\sigma\,\partial_t a^\lambda_i$, and the noncommutation relation \eqref{eq:noncomm} acquires $-\,d^{\lambda}_{\rho}\,\delta\pi^{\rho}$ on its right-hand side. Hamel's gear realizes the time-dependent analogue of the residue: two wheels with moments of inertia $I_1,I_2$ are coupled by the rheonomic constraint $\dot\varphi_1=f(t)\,\dot\varphi_2$ \cite[No.~240]{hamel:1949}. With the adapted coframe $\theta^{1}=\mathrm d\varphi_2$, $\theta^{2}=\mathrm d\varphi_1-f\,\mathrm d\varphi_2$, all structure functions vanish, $d^{2}_{1}=-\dot f$ and $P_2=I_1f\dot\varphi_2$, so the transport term $d^{2}_{1}P_2$ is exactly the term $-I_1f\dot f\,\dot\varphi_2$ that the equation formed from $T^{*}=\tfrac12\bigl(I_1f^{2}+I_2\bigr)\dot\varphi_2^{\,2}$ misses.

\begin{remark}[frame dependence; the orthogonal complement; reconciliation with the constrained-Lagrangian doctrine]\label{rem:orthogonal}
	The residue is not an invariant of the constrained system alone: the adapted frame fixes the coframe elements $\theta^s$ but not the complementary fields $e_s$, and $P_s$, the splitting of \eqref{eq:bh-embedded} and $\mathrm{Res}$ all refer to that choice. What is invariant is the restricted quadratic form \eqref{eq:Tplus}, since the admissible quasi-velocities are unchanged on $D$; what varies is how the discarded information is distributed between the constrained momenta and the structure functions.

	Each variant of Route B selects its complement: full elimination (B1) selects the coordinate complement $\partial_s$ (Proposition~\ref{prop:decomposition}). Choosing instead the complement $\mathcal M$-orthogonal to $D$ makes $\overline{\mathcal M}_{s\beta}=0$, hence $P_s\equiv0$ on $D$ and $\mathrm{Res}\equiv0$ identically---for every state and every applied force: with this choice the equations of the restricted energy \eqref{eq:reduced-op} alone are the correct equations of motion. Orthogonality is necessary and sufficient for $P_s\equiv0$; the residue can also vanish through $c^{s}_{\alpha\beta}=0$ on $D$ or by cancellation, so the orthogonal choice is the one that always works, not the only one. For the knife edge, $\tilde e_w=e_w-(ms/I_B)\,e_\vartheta$ achieves this (Figure~\ref{fig:frameplane}), and \eqref{eq:reduced-op} reproduces \eqref{eq:schneide-korrekt} verbatim. The restricted energy is now free of the mixed term; the offset enters through $I_B=I_S+ms^{2}$ and, in addition, through the structure functions of the new frame, $\tilde c^{\vartheta}_{\vartheta v}=ms/I_B$. The discarded information is not annihilated by the orthogonal choice; it is relocated from the momenta $P_s$ into the frame geometry.

	Three constructions of geometric mechanics are to be kept apart. The constrained Lagrangian---the restriction of $L$ to $D$ \cite{bloch:1996,bloch:2005,bloch:2015}---needs no complement at all. Its equations of motion in the form of Bloch, Marsden and Zenkov use the coordinate fields of the eliminated coordinates as complement---the vertical space of their Ehresmann connection---and the two correction terms of that form are those of Proposition~\ref{prop:decomposition}. The $\mathcal M$-orthogonal complement is a third object: it is the splitting behind the constrained connection of Lewis \cite{lewis:1996}, i.e., the Levi-Civita connection of the mass metric projected orthogonally onto $D$. None of this rehabilitates Route B: the operator \eqref{eq:reduced-op} still carries the structure functions $c^\beta_{\alpha\gamma}$ of a frame of $D$, and a \emph{coordinate} frame of $D$---which is what Lagrange's equations presuppose---exists exactly when $D$ is integrable. What is prohibited is not the restriction of the energy but the pretense that a coordinate calculus can process it.
\end{remark}

\subsection{Neighboring lines}\label{sec:geo-comparison}

Two further lines touch the question but not the debate. The Russian line, with full knowledge of the classical sources, develops Hadamard's and Hamel's conditions as the \emph{Hadamard--Hamel problem}, with constructive criteria and applications to wheeled vehicles \cite{borisov:2015}; in the present language, these conditions classify the constellations in which \eqref{eq:residue} vanishes for the equations one intends to keep.

A frame-based line has absorbed Hamel's equations. M\"uller \cite{muller:2021,muller:2023} develops \eqref{eq:bh} as a universal local-coordinate approach to constrained multibody and space systems, in which the Hamel coefficients are the structure functions \eqref{eq:cdef} and the choice of frame plays the role of the complement choice of Remark~\ref{rem:orthogonal}. Wensing \cite{wensing:2026} derives \eqref{eq:bh} from d'Alembert's principle via a Lie-bracket identity of the frame fields---the identity \eqref{eq:frameid}---and handles constraints expressed as vanishing generalized speeds by omitting the corresponding equations, which is the embedding \eqref{eq:bh-embedded}. Flannery's transpositional derivation also grounds Chetaev's rule for nonlinear velocity constraints \cite{flannery:2011}; its variational hypotheses remain under analysis \cite{talamucci:2025,talamucci:2026}. Extensions reach infinite-dimensional systems, field theories and structure-preserving integrators \cite{shi:2017,shi:2020,gao:2023}. These works share the algebra of Section~\ref{sec:resolution} and aim at formulation and computation; the legitimacy of the substitution is not posed as a question in them, and the classical record of Section~\ref{sec:ill-record} is absent.

Bloch's monograph \emph{Nonholonomic Mechanics and Control}, the standard reference of the geometric line, states the constrained-Lagrangian doctrine in its sharpest form: $L_c=L(q,\mathrm{hor}\,\dot q)$ is defined intrinsically and the curvature of the chosen Ehresmann connection is attached from the outset (Theorem~5.2.2); this is Proposition~\ref{prop:decomposition} in bundle language, with the coordinate complement as vertical space. A large collection of worked systems is developed toward control applications \cite[Chs.~3, 5]{bloch:2015,bloch:2005}. It is also the work of the modern line with the closest contact to the classical sources: Ferrers, Neumann (1888), Vierkandt, Korteweg, Chaplygin, Neimark--Fufaev and Rosenberg all appear. The substitution question, however, is not among the questions put to them. The episode of the 1890s is read as the variational-versus-d'Alembert alternative; Korteweg's critique is glossed as directed at ``the variational equations,'' although his targets had substituted the constraints into the \emph{vis viva} and applied Lagrange's formalism (Section~\ref{sec:ill-record}). Neumann appears with the article of 1888 but not with the paper of 1899, and Hadamard and the admissibility program are absent. The legitimacy question is thus answered in substance---the correction is never dropped---but not posed as a question.

What the present account adds to this literature is interpretation and demarcation rather than algebra: the identification of the substituted function as the restricted mass metric \eqref{eq:Tplus}; the resulting dichotomy---induced metric of a leaf, or fiberwise energy on $D$ that no longer determines the dynamics by itself---with Frobenius integrability as the discriminator inherited from Hadamard \cite[Nos.~5, 8]{hadamard:1895}; the residue formula \eqref{eq:residue} as the common source of the classical conditions; and the connection of these results to the record in which they originated. In this form---the restriction of a metric, one integrability test, and one formula for what is lost---the result is teachable.

\section{Discussion and conclusion}\label{sec:discussion}

Three limitations of scope apply. First, the analysis concerns the d'Alembertian route to the equations of motion. The variational question for nonholonomic systems, in which sense Hamilton's principle survives and how vakonomic dynamics differs, is a separate one with its own literature: H\"older's delimitation \cite{holder:1896}, Hamel's reading \cite[No.~105]{hamel:1949}, the realization problem \cite{kozlov:1983} and the experimental adjudication in favor of the d'Alembertian equations for rolling contact \cite{lewis:1995}. Premature substitution into an action functional is a distinct pitfall of the same family; its diagnosis is due to Flannery \cite{flannery:2005,flannery:2011}. Second, the constraints treated here are linear (Pfaffian); nonlinear velocity constraints and servo constraints raise separate issues. Third, what the paper adds to the modern debate is the record and the classical conditions, not a correction of its mathematics: the fundamental-equation framework \cite{udwadia:1996} is not at issue; what the paper questions is only what the example is taken to show about the foundations (Section~\ref{sec:ill-rediscovery}).

The mechanism by which the result was lost can be stated briefly: the primary literature of the 1890s and 1900s is written in German and French, was not translated into English before the 2010s (except Chaplygin's report \cite{chaplygin:2002}), and circulates mainly as page scans without a text layer; its twentieth-century carriers are specialist treatises (Section~\ref{sec:ill-record}). The papers of the modern debate, citing six \cite{udwadia:2010} and seven \cite{chen:2013} references, validate one another, while the Russian literature, in which the classical sources remained in use, was not affected \cite{borisov:2015,borisov:2016}.

The operation called ``Hamel's paradox'' since 2010 has had a name since 1899: the illegitimate form of the kinetic energy. Its prohibition, its explanation and its limits of validity were established between 1892 and 1908, restated by Hamel in the book cited as the paradox's source, and kept in print by the specialist literature and by geometric mechanics. The geometric analysis separates three statements that the debate had fused. Universal legitimacy of the coordinate shortcut, for every mass metric and every applied force, holds exactly when the distribution is integrable and the passage to the integral manifold is completed. For a given system the shortcut can be exact without integrability, through metric decoupling, vanishing bracket components or cancellation between constrained directions. The first two are the classical admissibility conditions of Hadamard \cite[No.~7]{hadamard:1895} and Hamel \cite[\S8]{hamel:1904a}; the third lies outside them, so that Hamel's necessity claim does not hold, and the exact criterion of the residue replaces it. Within the frame calculus the premature substitution loses exactly one identifiable term: the residue of the constrained momenta relative to the chosen complement. The restricted energy itself is well defined on the constraint distribution, but outside the integrable case it does not determine the dynamics by itself.

There is no paradox. There is a classical result with known conditions, and this paper returns it, together with its geometric form, to a debate that had proceeded without it.

\backmatter

\bmhead{A note on sources and translations}

The German sources (Vierkandt, Korteweg, Neumann, Boltzmann, Voss, St\"ackel, Hamel) were read in the original. For the French and Russian sources the author has relied on English translations: on those prepared by D.~H. Delphenich and made freely available at \url{www.neo-classical-physics.info}, with the originals consulted alongside them where accessible, and, for Chaplygin's paper of 1897, on the translation in \emph{Regular and Chaotic Dynamics} \cite{chaplygin:2002}. Hadamard's memoir \cite{hadamard:1895}, reprinted in Appell's booklet \cite{appell:1899} and contained in its translation, is cited by its numbered paragraphs, which are common to the original, the reprint and the translation. The passage quoted from it is rendered from the French, and a few formulas of the translation differ from the \emph{M\'emoires} (thus $m(m+1)/2$ in No.~7 appears in the translation as $m(m-1)/2$).

\begin{appendices}

	\section{Notation}\label{app:notation}

	\begingroup
	\footnotesize
	\renewcommand{\arraystretch}{1.25}
	\begin{tabular}{@{}p{0.34\textwidth}p{0.60\textwidth}@{}}
		\hline
		$q^i$; $i,j$                                                                & configuration coordinates; coordinate indices, $1,\dots,n$                                                  \\
		$\partial_i$, $\mathrm dq^i$                                                & coordinate vector fields $\partial/\partial q^i$ and coordinate one-forms; dual bases of $T_qQ$, $T_q^{*}Q$ \\
		$m$; $f=n-m$                                                                & number of Pfaffian constraints; degrees of freedom                                                          \\
		$B_{ri},\,B_r$                                                              & constraint coefficients, \eqref{eq:pfaff}; $r=1,\dots,m$                                                    \\
		$T,\ T^{*}$                                                                 & kinetic energy; its illegitimate form \eqref{eq:Tstar}                                                      \\
		$\EL{i}$                                                                    & Euler--Lagrange operator, \eqref{eq:dalembert}                                                              \\
		$Q_i,\ \overline Q_\sigma$                                                  & applied forces, coordinate resp.\ frame components                                                          \\
		$\mathcal M_{ij}$; $D$                                                      & mass metric \eqref{eq:massmetric}; constraint distribution \eqref{eq:distribution}                          \\
		$R$                                                                         & balance one-form \eqref{eq:balance}                                                                         \\
		$e_\lambda=b^i_\lambda\partial_i$; $\theta^\lambda=a^\lambda_i\mathrm dq^i$ & frame; dual coframe ($a^\lambda_ib^i_\mu=\delta^\lambda_\mu$)                                               \\
		$\lambda,\mu,\rho,\sigma$                                                   & frame indices, $1,\dots,n$                                                                                  \\
		$\alpha,\beta,\gamma$; $s,s'$                                               & admissible frame indices, $1,\dots,f$; constrained, $f{+}1,\dots,n$                                         \\
		$\omega^\lambda,\ \delta\pi^\lambda$                                        & quasi-velocities, quasi-displacements \eqref{eq:quasi}                                                      \\
		$c^\lambda_{\sigma\rho}$                                                    & structure functions \eqref{eq:cdef} (Hamel coefficients)                                                    \\
		$\overline T,\ \overline{\mathcal M}_{\lambda\mu},\ p_\lambda$              & frame kinetic energy, metric components, momenta \eqref{eq:framemomenta}                                    \\
		$\overline T^{+}$; $P_s$                                                    & restricted energy \eqref{eq:Tplus}; constrained momenta \eqref{eq:Ps}                                       \\
		$\mathrm{Res}_\alpha$                                                       & residue of the illegitimate reduction, \eqref{eq:residue}                                                   \\
		$K_{\alpha\beta\gamma}$                                                     & coefficient of the residue, \eqref{eq:Kcoeff}                                                               \\
		$d^{\lambda}_{\sigma}=b^i_\sigma\partial_t a^\lambda_i$                     & transport coefficients of a time-dependent frame (Section~\ref{sec:geo-decomposition})                      \\
		\hline
	\end{tabular}
	\endgroup

	\section{Proof of the frame identity \eqref{eq:frameid}}\label{app:frameid}

	With $T(q,\dot q)=\overline T(q,\omega(q,\dot q))$ and $\omega^\lambda=a^\lambda_i\dot q^i$,
	\[
		\frac{\partial T}{\partial\dot q^i}=p_\lambda\,a^\lambda_i ,
		\qquad
		\frac{\partial T}{\partial q^i}
		=\Bigl(\frac{\partial\overline T}{\partial q^i}\Bigr)_{\!\omega}
		+p_\lambda\,\partial_i a^\lambda_j\,\dot q^j .
	\]

	Hence
	\[
		\EL{i}(T)
		=\dot p_\lambda\,a^\lambda_i
		+p_\lambda\bigl(\partial_j a^\lambda_i-\partial_i a^\lambda_j\bigr)\dot q^j
		-\Bigl(\frac{\partial\overline T}{\partial q^i}\Bigr)_{\!\omega} ,
	\]
	and contraction with $b^i_\sigma$, using \eqref{eq:cdef} and $\dot q^j=b^j_\rho\omega^\rho$, gives
	\[
		b^i_\sigma\,\EL{i}(T)
		=\dot p_\sigma
		+c^\lambda_{\sigma\rho}\,\omega^\rho\,p_\lambda
		-e_\sigma(\overline T). \qquad\square
	\]

	\section{Hamel's \S8 conditions in modern notation}\label{app:hamel-par8}

	Hamel \cite[\S8]{hamel:1904a} asks when the Lagrange equations of the free coordinates, formed from the illegitimate $T$, are valid, and answers: the constraints must admit, by linear recombination, a splitting into two groups. For the first group, ``the energy [is] composed of two separate terms: the part that is free of [the constrained quasi-velocities] and a part that includes only those''; for the second group, the brackets of the admissible infinitesimal transformations contain no components along it. In the extreme case in which the second group is everything, ``the first $n-\nu$ infinitesimal transformations must commute with each other,'' and the constraints are integrable. Hamel calls the conditions necessary and sufficient; they contain Hadamard's theorem \cite[Nos.~7--8]{hadamard:1895} as the coordinate-velocity case. The sufficiency is correct. The necessity rests on his equation (8), which requires the unsymmetrized coefficient $\sum_{\rho>n-\nu}\beta_{\lambda\mu\rho}\,\partial^{2}T/\partial\omega_\rho\partial\omega_\kappa$ to vanish for all $\lambda,\mu,\kappa$, whereas the substituted equations require only the part symmetric in $\mu\kappa$ to vanish. The two agree for a single constraint; Hamel's extension to several constraints by recombination does not cover the symmetric cancellation of Section~\ref{sec:geo-residue}.

	For a single constraint the necessity is elementary. With one constrained index the criterion reads
	\[
		c_{\alpha\beta}\,\overline{\mathcal M}_{\gamma}+c_{\alpha\gamma}\,\overline{\mathcal M}_{\beta}=0
		\quad\text{for all }\alpha,\beta,\gamma ,
		\qquad\text{and for }\gamma=\beta:\quad 2\,c_{\alpha\beta}\,\overline{\mathcal M}_{\beta}=0 ;
	\]
	so every direction $\beta$ with $\overline{\mathcal M}_\beta\neq0$ has $c_{\alpha\beta}=0$ for all $\alpha$; inserting one such direction into the criterion for arbitrary $\alpha,\gamma$ leaves $c_{\alpha\gamma}\overline{\mathcal M}_\beta=0$, hence $c\equiv0$. At each configuration, therefore, either the coupling or the bracket vanishes entirely. The alternative is decided pointwise; on a region where the same alternative holds throughout, one mechanism acts.

	In the notation of Section~\ref{sec:geo-residue}, Hamel's alternative corresponds to the direction-by-direction vanishing of the residue \eqref{eq:residue}: for each constrained direction $e_s$, either
	\begin{align*}
		 & \overline{\mathcal M}_{s\beta}=0
		 &                                  & (\text{metric decoupling: the energy splits}), \\
		\text{or}\qquad
		 & c^{s}_{\alpha\beta}\big|_{D}=0
		 &                                  & (\text{no Frobenius component along } e_s),
	\end{align*}
	and Hamel's group sizes $\tau$, $\nu-\tau$ count the directions of the first and the second kind. Direction-by-direction vanishing is sufficient for $\mathrm{Res}\equiv0$, and necessary for a single constraint. Hamel's claim is that it is also necessary, after admissible linear recombination, for several constraints; the cancellation examples of Appendix~\ref{app:examples-cancel} refute this, and the example with the coordinate complement and the $\eta$-metric lies within the hypotheses of his theorem. For a single recombined direction, the second mechanism is Hadamard's criterion for the corresponding combination of constraints: his conditions $P_{i,h}=0$ are, up to sign and in coordinate velocities, the components of $\sum_s\lambda_s c^{s}_{\alpha\beta}|_D$ for the combination with multipliers $\lambda_s$, and his characterization of a substitutable combination as a differential that is exact at each point on the constraint surface is the vanishing of the restricted exterior derivative \cite[Nos.~7, 12]{hadamard:1895}. Since $c^{s}_{\alpha\beta}|_D$ are the components of the Frobenius tensor of $D$, the case in which the second mechanism acts for \emph{all} constrained directions is integrability---Hadamard's ``all constraints substitutable iff integrable'' \cite[No.~8]{hadamard:1895}.

	\section{Remark on Chen's criterion}\label{app:chen}

	Chen \cite{chen:2013} analyzes the embedding within the fundamental-equation (Udwadia--Kalaba) representation \cite{udwadia:1996,udwadia:2002}. His embedded function retains all coordinates as arguments and eliminates $m$ velocities, and his embedded equations are those of the retained velocities only: variant (B1) of Section~\ref{sec:ill-operation}, with affine and rheonomic constraints admitted. His criterion (4.1) compares the dynamics generated by the partitioned mass matrix of the substituted system with the projected dynamics of the constrained one and is necessary and sufficient for the embedded equations to be correct. In the present language, for homogeneous scleronomic constraints, it is the vanishing of the entire right-hand side of \eqref{eq:decomposition}---coordinate-slot term and residue of the coordinate complement together---read in the Udwadia--Kalaba representation. It is not the statement $\mathrm{Res}=0$: the example $\dot y=y\dot x$ of Proposition~\ref{prop:leaf} has $\mathrm{Res}=0$ and fails the criterion through the coordinate-slot term alone.

	Chen confirms the nonholonomic half of the Rosenberg conjecture on the constraint $\dot y=z\dot x$; his three holonomic cases, $y-kx=0$, $xy=r(t)$ and $x^{2}+y^{2}=\rho^{2}(t)$, all pass the criterion, and his closing remark asserts failure for the nonholonomic constraint only. The claim that holonomic embedding, too, can fail is Wanichanon and Cho's reading of the paper \cite{wanichanon:2024}, which they answer by the complete embedding of Proposition~\ref{prop:leaf}(a). The cases in which an integrable constraint does fail are those of incomplete embedding, Proposition~\ref{prop:leaf}(b).

	Chen's paper also contrasts the diagnoses of \cite{rosenberg:1977} (``[an] incorrect constraint was embedded'') and \cite{udwadia:2010} (conflation of the unconstrained system with the constraint); by Sections~\ref{sec:geo-frames}--\ref{sec:geo-decomposition} both are projections of the same fact: the restriction \eqref{eq:Tplus} discards the constrained momenta that the balance one-form still requires.

	\section{Examples for the criterion}\label{app:examples}

	\subsection{The coordinate-slot term and the residue in two classical examples}\label{app:examples-pair}

	The constraint $\dot y=y\,\dot x$ on a free particle in $\mathbb R^{3}$ (Proposition~\ref{prop:leaf}) is integrable, with integral $y=Ce^{x}$. The substituted energy $T^{*}=\tfrac12m\bigl[(1+y^{2})\dot x^{2}+\dot z^{2}\bigr]$ retains $y$, and the Lagrange equation in $x$ with $y$ held fixed gives $m(1+y^{2})\ddot x+2my^{2}\dot x^{2}=0$ along the constrained motion, whereas the correct equation, from Maggi's form $\langle R,\partial_x+y\,\partial_y\rangle=0$ or from the Lagrange equation on the leaf, is $m(1+y^{2})\ddot x+my^{2}\dot x^{2}=0$. The difference is the coordinate-slot term $y\,\partial T^{*}\!/\partial y$ of Proposition~\ref{prop:decomposition}; the residue vanishes, since the adapted frame fields commute.

	For the knife edge with $s=0$ the decomposition \eqref{eq:decomposition} reproduces the discrepancy of Section~\ref{sec:ill-schneide} quantitatively: $\partial T^{*}\!/\partial y=0$, $P_y=m\tan\vartheta\,\dot x$, $c^{y}_{x\vartheta}=-\sec^{2}\!\vartheta$, so the full-elimination equations err by $m\sec^{2}\!\vartheta\tan\vartheta\,\dot\vartheta\dot x$ in the $x$-equation and by $-m\sec^{2}\!\vartheta\tan\vartheta\,\dot x^{2}$ in the $\vartheta$-equation, although the frame-level shortcut (B3) is exact there.

	For Rosenberg's fallacy \cite[Ch.~14]{rosenberg:1977}, the free particle with $\dot y=z\dot x$, the formula gives $\partial T^{*}\!/\partial y=0$, $P_y=mz\dot x$ and $c^{y}_{xz}=-1$, and returns the two errors of his shortcut: $+mz\dot x\dot z$ in the $x$-equation, which doubles the coupling term, and $-mz\dot x^{2}$ in the $z$-equation, the spurious force. Both are the residue of the coordinate complement, since $T^{*}$ is free of $y$.

	\subsection{Cancellation between constrained directions}\label{app:examples-cancel}

	The smallest example has three admissible and three constrained directions. On $\mathbb R^6$ with coordinates $(x^1,x^2,x^3,z^1,z^2,z^3)$, unit mass and Euclidean kinetic energy, impose the three catastatic Pfaffian constraints $\theta^{3+t}:=\mathrm dz^t-\tfrac12\varepsilon_{tjk}\,x^j\mathrm dx^k=0$, $t=1,2,3$ (Levi-Civita symbol $\varepsilon$): each $z^t$ grows with the area swept by the projection of the motion onto one coordinate plane. Together with $\theta^\alpha=\mathrm dx^\alpha$ they form an adapted coframe; its dual frame is $e_{3+t}=\partial_{z^t}$ and $e_\alpha=\partial_{x^\alpha}-\tfrac12\varepsilon_{t\alpha j}\,x^j\partial_{z^t}$, and \eqref{eq:cdef} gives $c^{3+t}_{\alpha\beta}=-\mathrm d\theta^{3+t}(e_\alpha,e_\beta)=\varepsilon_{t\alpha\beta}$: every constrained direction receives one bracket, the bracket map has full rank, and the coupling $\overline{\mathcal M}_{3+t,\gamma}=-\tfrac12\varepsilon_{t\gamma j}x^j$ does not vanish either. Replacing the complement by $\tilde e_{3+t}=e_{3+t}+\kappa^\beta_t e_\beta$ leaves the constraint forms and the $c^{3+t}_{\alpha\beta}$ unchanged and turns the coupling into $\overline{\mathcal M}_{3+t,\gamma}+\kappa^\beta_t\overline{\mathcal M}_{\beta\gamma}$; since $\overline{\mathcal M}_{\beta\gamma}$ is invertible, $\kappa$ can be chosen so that the new coupling is $\eta\,\delta_{t\gamma}$ with any $\eta\neq0$. Then $K_{\alpha\beta\gamma}=\eta\,\varepsilon_{\gamma\alpha\beta}$ is antisymmetric in $\beta\gamma$, so $\mathrm{Res}\equiv0$ for every state, although no constrained direction is decoupled and none is free of brackets. This is a statement about the frame form (B3) with the chosen complement.

	The cancellation does not depend on the change of complement. Keep instead the coordinate complement $e_{3+t}=\partial_{z^t}$ and replace the Euclidean energy by the frame energy $\overline T=\tfrac12\sum_\alpha(\omega^\alpha)^{2}+\eta\sum_t\omega^{t}\omega^{3+t}+\tfrac12\sum_t(\omega^{3+t})^{2}$, $0<|\eta|<1$, a positive definite form with eigenvalues $1\pm\eta$. This gives the coupling $\eta\,\delta_{t\gamma}$ from the start, the same structure functions, again $K=\eta\,\varepsilon_{\gamma\alpha\beta}$, and a restricted energy $\overline T^{+}=\tfrac12\sum_\alpha(\dot x^\alpha)^{2}$ free of every coordinate. The Lagrange equations of $\overline T^{+}$ in the retained coordinates, $\ddot x^\alpha=0$, are then the correct equations of motion, by \eqref{eq:reduced-op} with vanishing residue and vanishing admissible structure functions, or directly from Maggi's form, although neither mechanism acts for any constrained direction. This is the setting of Hamel's theorem: the coordinate complement is fixed, and the abbreviated energy is free of the eliminated coordinates.

\end{appendices}

\bibliography{paper_hamel}% gemeinsame Bibliographie mit der RCD-Fassung

\end{document}